%% file: main.tex
\documentclass[a4paper]{article}
\usepackage{amsmath,amssymb} 
\usepackage{graphicx}
\usepackage{xspace}
\usepackage{comment}
\usepackage{times}
\usepackage{todonotes}
\usepackage{tikz}
\usepackage{rotating}

\usetikzlibrary{automata,positioning}

\renewcommand{\vec}[1]{\ensuremath{\mathbf{#1}}}

\newcommand\IN{{\mathbb N}}
\newcommand\IZ{{\mathbb Z}}

\newcommand{\integers}{\IZ}
\newcommand{\contractname}{\mathsf{Cnt}}
\newcommand{\contractstarname}{\ensuremath{\contractname^*}}
\newcommand{\contractstar}[1]{\ensuremath{\contractstarname\left(#1\right)}}
\newcommand{\contract}[1]{\ensuremath{\contractname\left(#1\right)}}
\newcommand{\contracti}[2]{\ensuremath{\contractname^{#2}\left(#1\right)}}

\newcommand{\effect}[1]{\ensuremath{\mathsf{Effect}\left(#1\right)}}
\newcommand{\src}[1]{\ensuremath{\mathsf{src}\left(#1\right)}}
\newcommand{\trg}[1]{\ensuremath{\mathsf{trg}\left(#1\right)}}

\newcommand{\tpof}[1]{\ensuremath{\mathsf{TPath}\left(#1\right)}}
\newcommand{\rof}[2]{\ensuremath{\mathsf{Run}\left(#1,#2\right)}}

\newcommand{\loc}{\ensuremath{\mathrm{Loc}}}
\newcommand{\rmax}{\ensuremath{\mathrm{rmax}}}
\newcommand{\cmax}{\ensuremath{\mathrm{cmax}}}
\newcommand{\edges}{\ensuremath{\mathrm{Edges}}}
\newcommand{\rate}{\ensuremath{\mathrm{Rates}}}
\newcommand{\invariants}{\ensuremath{\mathrm{Inv}}}
\newcommand{\init}{\ensuremath{\mathrm{Init}}}
\newcommand{\initval}{\ensuremath{\vec{0}}}

\newcommand{\goal}{\ensuremath{\mathrm{Goal}}}

\newcommand{\rateon}[1]{\ensuremath{\mathcal{R}\left(#1\right)}}

\newcommand{\guardson}[1]{\ensuremath{\mathcal{G}\left(#1\right)}}
\newcommand{\duration}[1]{\ensuremath{\mathsf{duration}\left(#1\right)}}
\newcommand{\first}[1]{\ensuremath{\mathsf{first}\left(#1\right)}}
\newcommand{\last}[1]{\ensuremath{\mathsf{last}\left(#1\right)}}

\newcommand{\dest}[1]{\ensuremath{\mathsf{dest}\left(#1\right)}}
\newcommand{\len}[1]{\ensuremath{\left|#1\right|}}

\newcommand{\cH}{\ensuremath{\mathcal{H}}}
\newcommand{\posreal}{\ensuremath{\mathbb{R}^+}}
\newcommand{\reals}{\ensuremath{\mathbb{R}}}

\newcommand{\true}{\ensuremath{\mathbf{true}}}
\newcommand{\tb}{\ensuremath{\mathbf{T}}}

\newcommand{\val}{\ensuremath{\nu}}
\newcommand{\tsucc}[1]{\ensuremath{\leq_{\mathrm{ts}}^{#1}}}
\newcommand{\valL}[1]{\ensuremath{{\sf Val}^{\leftarrow}({#1})}}
\newcommand{\valR}[1]{\ensuremath{{\sf Val}^{\rightarrow}({#1})}}
\newcommand{\sem}[1]{\ensuremath{[ \! [ {#1} ] \! ]}}
\newcommand{\Pre}{{\sf Pre}} \newcommand{\Post}{{\sf Post}}
\newcommand{\pre}{\ensuremath{{\sf pre}^\sharp}}
\newcommand{\post}{\ensuremath{{\sf post}^\sharp}}
\newcommand{\pathname}{\ensuremath{\pi_{\contract{\rho}}}}
\newcommand{\nexptm}{NExpTM\xspace}

\newtheorem{theorem}{Theorem}
\newtheorem{lemma}{Lemma}
\newtheorem{corollary}{Corollary}
\newtheorem{proposition}{Proposition}
\newtheorem{problem}{Problem}

\newcommand{\nat}{\IN}
\renewcommand{\H}{\mathcal{H}}

\newcommand{\ellbar}{\ensuremath{\overline{\ell}}}
\newcommand{\zplus}{\ensuremath{\mathbf{0}^+}}
\newcommand{\zeq}{\ensuremath{\mathbf{0}^=}}

\newcommand{\Reg}[1]{\ensuremath{{\sf Reg}\left(#1\right)}}
\newcommand{\RegHA}[1]{\ensuremath{{\sf R}\left(#1\right)}}

\newcommand{\regof}[1]{\ensuremath{\left[#1\right]}}
\newcommand{\shadp}{SHA\ensuremath{{}^{\geq 0}}\xspace}
\newcommand{\rhadp}{RHA\ensuremath{{}^{\geq 0}}\xspace}
\newcommand{\fpart}[1]{\ensuremath{\left\langle #1\right\rangle}}
\newcommand{\ipart}[1]{\ensuremath{\left\lfloor #1\right\rfloor}}

\newcommand{\reach}[1]{\ensuremath{{\sf Reach}^{\leq #1}}}
\newcommand{\coreach}[1]{\ensuremath{{\sf coReach}^{\leq #1}}}

\newenvironment{proof}{\noindent{\it Proof.}\hspace*{.5cm}}{\hfill $\Box$}

\title{Time-bounded Reachability for Hybrid Automata:\\ Complexity and
  Fixpoints}

\author{Thomas Brihaye${}^\star$ \and Laurent Doyen${}^+$ \and Gilles
  Geeraerts${}^\%$ \and Jo\"el Ouaknine${}^\dagger$ \and
  Jean-Fran\c{c}ois Raskin${}^\%$ \and James Worrell${}^\dagger$}

\date{$\star$: Universit\'e de Mons, Belgium\\ $+$: LSV, ENS Cachan
  \& CNRS, France\\
  $\%$: Universit\'e Libre de Bruxelles, Belgium\\ $\dagger$: Oxford
  University Computing Lab., UK}

\begin{document}

\maketitle





\input{abstract}

\input{introduction}
\input{definitions}
\input{pathcompression}
\input{algorithm}
\input{pract-contrib}

\input{main.bbl}
\appendix

\input{appendix}

\end{document}

%% file: abstract.tex
\begin{abstract}
  In this paper, we study the\emph{ time-bounded reachability problem}
  for rectangular hybrid automata with non-negative rates (\rhadp). This
  problem was recently shown to be decidable ~\cite{ICALP11} (even
  though the \emph{unbounded} reachability problem for even very
  simple classes of hybrid automata is well-known to be
  undecidable). However,~\cite{ICALP11} does not provide a precise
  characterisation of the complexity of the time-bounded reachability
  problem. The contribution of the present paper is threefold. First,
  we provide a new \textsc{NExpTime} algorithm to solve the
  timed-bounded reachability problem on \rhadp. This algorithm
  improves on the one of~\cite{ICALP11} by at least one
  exponential. Second, we show that this new algorithm is optimal, by
  establishing a matching lower bound: time-bounded reachability for
  \rhadp is therefore \textsc{NExpTime}-complete. Third, we extend
  these results in a practical direction, by showing that we can
  effectively compute fixpoints that characterise the sets of states
  that are reachable (resp. co-reachable) within $T$ time units from a
  given starting state.


\end{abstract}


%% file: introduction.tex
\section{Introduction}
\label{sec:introduction}
\emph{Hybrid systems} form a general class of systems that mix
\emph{continuous} and \emph{discrete} behaviors. Examples of hybrid
systems abound in our everyday life, particularly in applications where
an (inherently discrete) computer system must interact with a
continuous environment. The need for modeling hybrid systems is
obvious, together with methods to analyse those systems.

\emph{Hybrid automata} are arguably among the most prominent families
of models for hybrid systems
\cite{DBLP:conf/lics/Henzinger96}. Syntactically, a hybrid automaton
is a finite automaton (to model the discrete part of the system)
augmented with a finite set of real-valued variables (to model the
continuous part of the system). Those variables evolve with time
elapsing, at a rate which is given by a flow function that depends on
the current location of the automaton.  The theory of hybrid automata
has been well developed for about 20 years, and tools to analyse them
are readily available, see for instance \textsc{Hytech}
\cite{HHW95,DBLP:conf/cav/HenzingerHW97} and \textsc{Phaver}
\cite{Frehse08}.

Hybrid automata are thus a class of powerful models, yet their high
expressiveness comes at a price, in the sense that the undecidability
barrier is rapidly hit. Simple \emph{reachability properties} are
undecidable even for the restricted subclass of \emph{stopwatch
  automata}, where the rate of growth of each variable stays constant
in all locations and is restricted to either $0$ or $1$ (see
\cite{HenzingerKPV98} for a survey).

On the other hand, a recent and successful line of research in the
setting of \emph{timed automata} has outlined the benefits of
investigating \emph{timed-bounded variants} of classical properties
\cite{DBLP:conf/lics/JenkinsORW10,DBLP:conf/concur/OuaknineRW09}. For
instance, while \emph{language inclusion} is, in general undecidable
for timed automata, it becomes decidable when considering only
executions of \emph{bounded duration}
\cite{DBLP:conf/concur/OuaknineRW09}.

In a recent work \cite{ICALP11} we have investigated the decidability
of \emph{time-bounded reachability} for rectangular hybrid automata
(i.e., is a given state reachable by an execution of duration at most
$T$ ? for a given $T$). We have shown that \emph{time-bounded}
reachability is \emph{decidable} for \emph{rectangular hybrid automata
  with non-negative rates} (\rhadp), while it is well-known
\cite{HenzingerKPV98} that (plain, time unbounded) reachability is not
for this class. We have also shown that the decidability frontier is
quite sharp in the sense that time-bounded reachability becomes
\emph{undecidable} once we allow either \emph{diagonal constraints} in
the guards or \emph{negative rates}.

To obtain decidability of time-bounded reachability for \rhadp, we
rely, in \cite{ICALP11}, on a \emph{contraction operator} that applies
to runs, and allows to derive, from any run of duration at most $T$ of
an \rhadp $\cH$, an equivalent run that reaches the \emph{same
  state}, but whose length (in terms of number of discrete
transitions) is \emph{uniformly bounded} by a function $F$ of the size
of $\cH$ and $T$. Hence, deciding reachability within $T$ time units
reduces to exploring runs of bounded lengths only, which is feasible
algorithmically (see \cite{ICALP11} for the details). However, this
previous work does not contain a precise characterisation of the
complexity of time-bounded reachability. Clearly, an upper bound on
the complexity depends on the bound $F$ on the length of the runs that
need to be explored.

In the present work, we revisit and extend our previous results
\cite{ICALP11} in several directions, both from the theoretical and
the practical point of view. \emph{First}, we completely revisit the
definition of the \emph{contraction operator} and obtain a new
operator that allows to derive a \emph{singly exponential upper bound}
on the lengths of the runs that need to be considered, while the
operator in \cite{ICALP11} yields an upper bound that is at least
doubly exponential. Our new contraction operator thus provides us with
an \textsc{NExpTime} algorithm that improves on the algorithm of
\cite{ICALP11} by at least one exponential. \emph{Second}, we show
that this new algorithm is optimal, by establishing a matching lower
bound. Hence, \emph{time-bounded reachability for \rhadp is
  \textsc{NExpTime}-complete}. \emph{Third}, we extend those results
towards more practical concerns, by showing that we \emph{can
  effectively compute fixpoints} that characterise the set of states
that are reachable (resp. co-reachable) within $T$ time units, from a
given state. The time needed to compute them is at most doubly
exponential in the size of the \rhadp and the bound
$T$. \emph{Fourth}, \emph{we apply those ideas to two examples} of
\rhadp for which the classical (time-unbounded) forward and backward
fixpoints do not terminate. We show that, in those examples, the sets
of states that are time-bounded reachable is computable in practice,
for values of the time bound that allow us to derive \emph{meaningful
  properties}.

This brief summary of the results outlines the structure of the
paper. Remark that, by lack of space, some more technical proofs have
been moved to the appendix.


%% file: definitions.tex
\section{Definitions}

Let ${\cal I}$ be the set of intervals of real numbers with endpoints in
$\integers\cup\{-\infty,+\infty\}$.
Let $X$ be a set of continuous variables, and let $\dot{X} = \{\dot{x}
\mid x \in X\}$ be the set dotted variables, corresponding to variable
first derivatives.  A \emph{rectangular constraint} over $X$ is an
expression of the form $x\in I$ where $x$ belongs to $X$ and $I$ to
${\cal I}$.
A \emph{diagonal constraint} over $X$ is a constraint of the form
$x-y\sim c$ where $x,y$ belong to $X$, $c$ to $\integers$, and $\sim$
is in $\{ <, \le , =, \ge, > \}$.
Finite conjunctions of diagonal and rectangular constraints over $X$
are called {\em guards}, over $\dot{X}$ they are called {\em rate
  constraints}.  A guard or rate constraint is {\em rectangular} if
all its constraints are rectangular.  We denote by $\guardson{X}$ and
$\rateon{X}$ respectively the sets of guards and rate constraints over
$X$.


\paragraph{{Linear, rectangular and singular hybrid automata}}
A \emph{linear hybrid automaton} (LHA) is a tuple ${\mathcal{H}} = (X,
\loc, \edges, \rate, \invariants, \init)$ where
$X=\{x_1,\ldots, x_{|X|}\}$ is a finite set of continuous
\emph{variables} ; $\loc$ is a finite set of \emph{locations}; $\edges
\subseteq \loc \times \guardson{X} \times 2^X \times \loc$ is a finite
set of \emph{edges}; $\rate:\loc \mapsto\rateon{X}$ assigns to each
location a constraint on the \emph{possible variable rates};
$\invariants: \loc\mapsto\guardson{X}$ assigns an \emph{invariant} to
each location; and $\init\subseteq \loc$ is a \emph{set of initial
  locations}. For an edge $e=(\ell, g, Y, \ell')$, we denote by
$\src{e}$ and $\trg{e}$ the location $\ell$ and $\ell'$ respectively,
$g$ is called the \emph{guard} of $e$ and $Y$ is the \emph{reset} set
of $e$. In the sequel, we denote by $\rmax$ and $\cmax$ the maximal
constant occurring respectively in the constraints of
$\{\rate(\ell)\mid \ell\in\loc\}$ and of $\{\rate(\ell)\mid
\ell\in\loc\}\cup\{g\mid \exists (\ell, g, Y, \ell')\in \edges\}$.

An LHA is \emph{non-negative rate} if for all variables $x$, for all
locations $\ell$, the constraint $\rate(\ell)$ implies that $\dot{x}$
must be non-negative.  A \emph{rectangular hybrid automaton} (RHA) is
a linear hybrid automaton in which all guards, rates, and invariants
are rectangular. In the case of RHA, we view rate constraints as
functions $\rate : \loc \times X \rightarrow {\cal I}$ that associate
with each location $\ell$ and each variable $x$ an interval of
possible rates $\rate(\ell)(x)$.  A \emph{singular hybrid automaton}
(SHA) is an RHA s.t.  for all locations $\ell$ and for all variables
$x$: $\rate(\ell)(x)$ is a singleton. We use the shorthands \rhadp and
\shadp for \emph{non-negative rates} RHA and SHA respectively.

\paragraph{{LHA semantics}}
A \emph{valuation} of a set of variables $X$ is a function
$\val:X\mapsto\reals$.  We denote by $\initval$ the valuation that
assigns $0$ to each variable.

Given an LHA $\cH=(X,\loc,\edges,\rate,\invariants,\init, X)$, a
\emph{state} of $\cH$ is a pair $(\ell, \val)$, where $\ell\in \loc$
and $\val$ is a valuation of $X$. The semantics of $\cH$ is defined as
follows. Given a state $s=(\ell,\val)$ of $\cH$, an \emph{edge step}
$(\ell,\val)\xrightarrow{e}(\ell',\val')$ can occur and change the
state to $(\ell',\val')$ if $e=(\ell, g, Y, \ell') \in \edges$,
$\val\models g$, $\val'(x)=\val(x$) for all $x\not\in Y$, and
$\val'(x)=0$ for all $x\in Y$; given a time delay $t\in\posreal$, a
\emph{continuous time step} $(\ell,\val)\xrightarrow{t}(\ell,\val')$
can occur and change the state to $(\ell,\val')$ if there exists a
vector $r=(r_1,\ldots r_{|X|})$ such that $r\models\rate(\ell)$,
$\val'=\val+(r\cdot t)$, and $\val + (r \cdot t') \models
\invariants(\ell)$ for all $0 \leq t' \leq t$.

A \emph{path} in $\cH$ is a finite sequence $e_1, e_2, \ldots, e_n$ of
edges such that $\trg{e_i} = \src{e_{i+1}}$ for all $1\leq i\leq n-1$.
A \emph{timed path} of $\cH$ is a finite sequence of the form
$\pi=(t_1, e_1), (t_2,e_2),\ldots,(t_n, e_n)$, such that $e_1,\ldots,
e_n$ is a path in $\cH$ and $t_i\in\posreal$ for all $0\leq i\leq
n$. For all $k$, $\ell$, we denote by $\pi[k:\ell]$ the maximal
portion $(t_i,e_i), (t_{i+1}, e_{i+1}),\break \ldots, (t_j, e_j)$ of $\pi$
s.t. $\{i, i+1,\ldots, j\}\subseteq [k,\ell]$ (remark that the
interval $[k,\ell]$ could be empty, then $\pi[k:\ell]$ is empty too).
Given a timed path $\pi=(t_1, e_1), (t_2,e_2),\ldots,(t_n, e_n)$ of an
SHA, we let $\effect{\pi} = \sum_{i=1}^n \rate(\ell_{i-1})\cdot t_{i}$
be the \emph{effect of $\pi$} (where $\ell_i=\src{e_i}$ for $1 \le i
\le n$).


A \emph{run} in $\cH$ is a sequence $s_0,
(t_1,e_1), s_1, (t_2, e_2),\ldots,  (t_{n}, e_{n}), s_n$ such that:
\begin{itemize}
\item $(t_1, e_1),(t_2,e_2),\ldots, (t_{n},e_{n})$ is a timed path in
$\cH$, and 
\item for all $0 \leq i < n$, there exists a state $s_i'$ of $\cH$
  with
  $s_i\xrightarrow{t_{i+1}}s_i'\xrightarrow{e_{i+1}}s_{i+1}$. 
\end{itemize}
Given a run $\rho=s_0, (t_1,e_1),\dots, s_n$, let $\first{\rho} = s_0
= (\ell_0, \val_0)$, $\last{\rho} = s_n$, $\duration{\rho} =
\sum_{i=1}^{n} t_i$, and $\len{\rho}=n+1$.  We say that~$\rho$ is
\emph{$\tb$-time-bounded} (for $\tb \in \IN$) if $\duration{\rho} \leq
\tb$. Given two runs $\rho=s_0,(t_1,e_1),\ldots, (t_n,e_n),s_n$ and
$\rho'=s_0',(t_1',e_1'),\ldots, (t_k',e_k'), s_k'$ with $s_n=s_0'$, we
let $\rho\cdot\rho'$ denote the run
$s_0,(t_1,e_1),\ldots,(t_n,e_n),s_n,(t_1',e_1'),\ldots, (t_k',e_k'),
s_k'$. 

Note that a unique timed path $\tpof{\rho}=(t_1,
e_1),(t_2,e_2),\ldots,\break (t_{n},e_{n})$, is associated with each
run $\rho=s_0, (t_1,e_1), s_1,\ldots,\break (t_{n}, e_{n}), s_n$.
Hence, we sometimes abuse notation and denote a run $\rho$ with
$\first{\rho}=s_0$, $\last{\rho}=s$ and $\tpof{\rho}=\pi$ by
$s_0\xrightarrow{\pi}s$. The converse however is not true: given a
timed path $\pi$ and an initial state $s_0$, it could be impossible to
build a run starting from $s_0$ and following $\pi$ because some
guards or invariants along $\pi$ might be violated. However, if such a
run exists it is necessarily unique \emph{when the automaton is
  singular}. In that case, we denote by $\rof{s_0}{\pi}$ the function
that returns the unique run $\rho$ such that $\first{\rho}=s_0$ and
$\tpof{\rho}=\pi$ if it exists, and $\bot$ otherwise. Remark that,
when consider an SHA: if
$\rho=(\ell_0,\val_0)\xrightarrow{\pi}(\ell_n,\val_n)$ is a run, then
for all $x$ that is \emph{not reset} along $\rho$:
$\val_n(x)=\val_0(x)+\effect{\pi}(x)$.

\paragraph{{Time-bounded reachability problem for LHA}}
While the reachability problem asks whether there exists a run
reaching a given goal location, we are only interested in runs having
\emph{bounded duration}.

\begin{problem}[Time-bounded reachability problem]\label{prob:tbr}
  Given an LHA ${\mathcal{H}} = (X,\loc,\break \edges,
  \rate,\invariants,\init)$, a location $\goal \in \loc$ and a time
  bound $\tb \in \IN$, the \emph{time-bounded reachability problem} is
  to decide whether there exists a finite run $\rho=(\ell_0, \initval)
  \xrightarrow{\pi} (\goal,\cdot)$ of ${\mathcal{H}}$ with
  $\ell_0\in\init$ and $\duration{\rho} \le \tb$.
\end{problem}

This problem is known to be decidable~\cite{ICALP11} for \rhadp, but
its exact complexity is, so far, unknown. We prove in
Section~\ref{sec:time-bound-reach} (thanks to the results of
Section~\ref{sec:contr-oper-timed}) that it is
\textsc{NExpTime}-complete. This problem is known to become
undecidable once we allow either diagonal constraints in the guards,
or negative and positive rates to occur in the LHA~\cite{ICALP11}.

A more general problem that is relevant in practice, is to compute a
symbolic representation of all the states that are reachable in at
most $\tb$ time units. Here, by `symbolic representation' we mean a
finite representation of the set of states that can be manipulated
algorithmically. This problem, together with the definition of such a
such a symbolic representation, will be addressed in
Section~\ref{sec:pract-impl}.

\input{simplex}




%% file: simplex.tex
Let us illustrate, by means of the \rhadp $\cH$ in Fig.~\ref{fig:simple
  example}, the difficulties encountered when computing the reachable
states of a \rhadp. Let us characterise the set $\ensuremath{{\sf
    Reach}}_{\ell_1}(s_0)$ of all states of the form $(\ell_1,\val)$
that are reachable from $s_0$. It is easy to see that
$\ensuremath{{\sf Reach}}_{\ell_1}(s_0)=\{(\ell_1, (0,
\frac{1}{2^n}))\mid n\in \IN_0\}$. Moreover, observe that, for all
$n\in \IN_0$, $(\ell_1, (0, \frac{1}{2^n}))$ is reachable from $s_0$
by one and only one run, of duration 
$(n-1)+\frac{1}{2^n}$, 
and that the
number of bits necessary to encode those states grows \emph{linearly}
with the length of the run. This examples shows that finding an
adequate, compact and effective representation (such as regions in the
case of Timed Automata~\cite{AD94}) for the set of reachable of an
\rhadp is not trivial (and, in full generality, impossible because
reachability is undecidable for this class). Nevertheless, in
Section~\ref{sec:pract-impl}, we show that, in an \rhadp, an effective
representation of the set of states that are reachable \emph{within
  $\tb$ time units} can be computed.



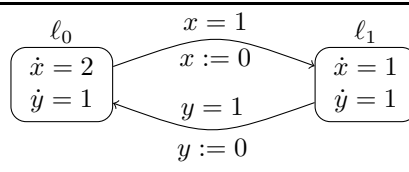
\begin{figure}
\hrule
\begin{center}
\begin{tikzpicture}
  \path (-4,0) node[draw,rectangle,rounded corners=2mm,inner sep=2pt]
  (q0) {$\begin{array}{c} \dot{x} = 2 \\ \dot{y} = 1  \end{array}$};

  \path (0,0) node[draw,rectangle,rounded corners=2mm,inner sep=2pt]
  (q1) {$\begin{array}{c} \dot{x} = 1 \\ \dot{y} =1  \end{array}$};

  \path (-4,.7) node[] (q0b) {$\ell_0$};
  \path (0,.7) node[] (q0b) {$\ell_1$};

   \draw[arrows=<-] (q0) .. controls +(340:2.1cm) .. (q1)
      node[pos=.5,above,sloped] {{$y=1$}}
      node[pos=.5,below,sloped] {${y:=0}$};

   \draw[arrows=<-] (q1) .. controls +(160:2.1cm) .. (q0)
      node[pos=.5,above,sloped] {{$x=1$}}
      node[pos=.5,below,sloped] {${x:=0}$};
\end{tikzpicture}
\end{center}
\hrule
\caption{A simple hybrid automaton.}
\label{fig:simple example}
\end{figure}


%% file: pathcompression.tex
\section{Contracting runs}
\label{sec:contr-oper-timed}
In this section, we describe a \emph{contraction operator}. Given an
RHA $\cH$, and one of its timed paths $\pi$ of arbitrary length but of
duration $\leq \tb$, the contraction operator builds a timed path
$\contractstar{\pi}$ that reaches the same state as $\pi$, but whose
size is uniformly bounded by a function of $\tb$, and of the size of
$\cH$. This operator is central to prove correctness of the algorithms
for time-bounded reachability in sections~\ref{sec:time-bound-reach}
and~\ref{sec:pract-impl}. Since Problem~\ref{prob:tbr} is undecidable
if both positive and negative rates are allowed \cite{ICALP11}, we
restrict our attention to RHA with non negative rates. Moreover, for
the sake of clarity, all the results presented in this section are
limited to \emph{singular hybrid automata}, but they extend easily to
\rhadp as we will see later.  Thus, from now one, we fix an \shadp
$\cH=(X, \loc, \edges, \rate, \invariants, \init)$.

\paragraph{Self loops} The first step of our construction consists in
adding, on each location $\ell$ of $\cH$, a self-loop $(\ell,
\true, \emptyset,\ell)$. The resulting \shadp is called $\cH'$. Those
self-loops allow to split runs of $\cH'$ into portions of arbitrary
small delays, because if $\cH'$ admits a run of the form $(\ell,\val),
(t_1+t_2, e), s$, it also admits the run $(\ell,\val), (t_1, e'),
(\ell,\val'), (t_2,e), s$, where $e'$ is the self loop on
$\ell$. Yet, this construction preserves (time-bounded)
reachability:
\begin{lemma}\label{lem:cH-cH-prime}
  Every run of $\cH$ is also a run of $\cH'$. Conversely, if $\cH'$
  admits a run $\rho'$ with $\first{\rho'}=s_1$ and
  $\last{\rho'}=s_2$, then $\cH$ admits a run $\rho$ with $\rho$ with
  $\first{\rho}=s_1$ $\last{\rho}=s_2$,
  $\duration{\rho}=\duration{\rho'}$ and $|\rho|\leq |\rho'|$.
  Moreover, for each run $\rho$ of $\cH'$, there exists a run
  $\rho'=\rho_1\cdot\rho_2\cdots\rho_n$ of $\cH'$ s.t. $n\leq
  \duration{\rho}\times\rmax+1$, $\first{\rho}=\first{\rho'}$,
  $\last{\rho}=\last{\rho'}$, $\duration{\rho}=\duration{\rho'}$ and,
  for all $1\leq i\leq n$: $\duration{\rho_i} < \frac{1}{\rmax}$.
\end{lemma}

\paragraph{Hybrid automaton with regions} Let us describe a second
construction that applies to the syntax of the hybrid automaton, and
consists, roughly speaking, in encoding the integral part of the
variable valuations in the locations.  Let $\Reg{\cmax}=\big(\{[a,a],\
(a-1,a)\mid a\in\{1,\ldots,\cmax\}\} \cup
\{\zeq,\zplus,(\cmax,+\infty)\}\big)^X$ be the set of \emph{regions},
and further let $\Reg{\cmax,X}$ denote the set of all functions
$r:X\mapsto\Reg{\cmax}$ that assign a region to each variable. By
abuse of language, we sometimes call \emph{regions} elements of
$\Reg{\cmax,X}$ too. Remark that the definition of $\Reg{\cmax,X}$
differs from the classical regions \cite{AD94} by the absence of
$[0,0]$ which is replaced by two symbols: $\zeq$ and $\zplus$, and by
the fact that no information is retained about the relative values of
the fractional parts of the variables. The difference between $\zeq $
and $\zplus$ is elucidated later (see
Lemma~\ref{lemma:value-variables}). When testing for membership to a
region, $\zplus$ and $\zeq$ should be interpreted as $[0,0]$, i.e.,
$v\in\zplus$ and $v\in\zeq$ hold iff $v=0$. Given a valuation $\val$
of the set of variable $X$, and $r\in\Reg{\cmax,X}$, we let $\val\in r$
iff $\val(x)\in r(x)$ for all $x$, and, provided that $\val>\initval$,
we denote by $\regof{\val}$ the (unique) element from $\Reg{\cmax, X}$
s.t. $\val \in \regof{\val}$. Remark that for all sets of variable $X$
and all maximal constants $\cmax$: $|\Reg{\cmax,X}|\leq
(2\times(\cmax+1))^{|X|}$. Let $r_1$ and $r_2$ be two regions in
$\Reg{\cmax,X}$, and let $v:X\mapsto\mathbb{R}$ be a function
assigning a rate $v(x)$ to each variable $x$. Then, we say that
\emph{$r_2$ is a time successor of $r_1$ under $v$} (written
$r_1\tsucc{v}r_2$) iff there are $\val_1\in r_1$, $\val_2\in r_2$ and
a time delay $t$ s.t. $\val_2=\val_1+t\cdot v$. Remark that, by this
definition, we can have $r_1\tsucc{v}r_2$, $r_1(x)=\zeq $ and
$r_2(x)=\zplus$ for some clock $x$ (for instance, if $v(x)=0$).
 
Let us now explain how we label the locations of $\cH'$ by regions. We
let $\RegHA{\cH'}=(X, \loc', \edges', \rate', \invariants', \init')$
be the \shadp where:
\begin{itemize}
\item $\loc'=\loc\times\Reg{\cmax,X}$
\item for all $(\ell,r)\in\loc'$: $\rate'(\ell,r)=\rate(\ell)$
\item for all $(\ell,r)\in\loc'$:
  $\invariants(\ell,r)=\invariants(\ell)\wedge\bigwedge_{x:r(x)=\zeq }
  {x=0}$
\item There is an edge $e'=\big((\ell,r), g\wedge x\in r''\wedge g_0,
  Y, (\ell',r')\big)$ in $\edges'$ iff there are an edge
  $e=(\ell,g,Y,\ell')$ in $\edges$ and a region $r''$ s.t.:
  $r\tsucc{\rate(\ell)} r''$, for all $x\not\in Y$: $r'(x)=r''(x)$,
  for all $x\in Y$: $r'(x)\in \{\zeq , \zplus\}$ and
  $g_0=\bigwedge_{x\in X} g_0(x)$ where:
  \begin{eqnarray}
    \forall x\in X &:& g_0(x)=\left\{
      \begin{array}{ll}
        x=0&\textrm{if }r(x)=\zeq \\
        x>0&\textrm{if }r(x)=\zplus\\
        \mathbf{true}&\textrm{otherwise}
      \end{array}
    \right.\label{eq:def-g0}
  \end{eqnarray}
  in this case, we say that $e$ is the (unique) edge of $\cH'$
  \emph{corresponding} to $e'$. Symmetrically, $e'$ is the only edge
  corresponding to $e$ between locations $(\ell,r)$ and $(\ell',r')$.
\item $\init'=\init\times \{\zeq ,\zplus\}^X$
\end{itemize}
It is easy to see that this construction incurs an exponential blow up
in the number of locations. More precisely:
\begin{align}
  |\loc'|&\leq |\loc|\times|\Reg{\cmax,X}| \nonumber \\
  &= |\loc|\times(2\times(\cmax+1))^{|X|}\label{eq:number-of-locations-reg-H}
\end{align}

Let us prove that this construction preserves reachability of states:
\begin{lemma}\label{lemma:reg-ha-preserves-reach} 
  Let $s=(\ell,\val)$ and $s'=(\ell',\val')$ be two states of
  $\cH'$. Then, $\cH'$ admits a run $\rho$ with $\first{\rho}=s$ and
  $\last{\rho}=s'$ iff there are $r$ and $r'$ s.t. $\RegHA{\cH'}$
  admits a run $\rho'$ with $\first{\rho'}=((\ell, r),\val)$,
  $\last{\rho'}=((\ell',r'),\val')$,
  $\duration{\rho}=\duration{\rho'}$ and $\len{\rho}=\len{\rho'}$.
\end{lemma}
%

Intuitively, the regions that label locations in $\RegHA{\cH'}$ are
intended to track the region to which each variable belongs when
entering the location. However, in the case where a variable $x$
enters a location with value $0$, we also need to remember whether $x$
is still null when crossing the next edge (for reasons that will be
made clear later). This explains why we have two regions, $\zeq $ and
$\zplus$, corresponding to value $0$. They encode respectively the
fact that the variable is null (strictly positive) when leaving the
location.

Formally, we say that a run
$\rho=((\ell_0,r_0),\val_0),(t_1,e_1),
((\ell_1,r_1),\val_1),\ldots, (t_n,e_n),\break((\ell_{n},r_{n}),\val_{n})$ of
$\RegHA{\cH'}$ is \emph{region consistent} iff $(i)$ for all $0\leq
i\leq n$: $\val_i\in r_i$ and $(ii)$ for all $0\leq i\leq n-1$, for
all $x\in X$: $r_i(x)=\zeq $ implies
$\val_i(x)+ \linebreak t_{i+1}\times\rate(\ell_i)(x)=0$ and $r_i(x)=\zplus$
implies $\val_i(x)+t_{i+1}\times\rate(\ell_i)(x)>0$. Then, it is easy
to see that the construction of $\RegHA{\cH}$ guarantees that all runs
are region consistent:

\begin{lemma}\label{lemma:value-variables}
  All runs of $\RegHA{\cH'}$ are region consistent.
\end{lemma}

The contraction operator we are about to describe preserves
reachability of states when applied to carefully selected run portions
only. Those portions are obtained by splitting several times a
complete run into sub-runs, that we categorise in 4 different
\emph{types}.

\paragraph{Type-0 and type-1 runs} The notion of \emph{type-0} run
relies on the fact that each $\tb$-time bounded run of $\cH'$ (hence
of $\RegHA{\cH'}$) corresponds to a run $\rho'$ that can be split into
at most $\tb\times\rmax+1$ portions of duration $<\frac{1}{\rmax}$
(see Lemma~\ref{lem:cH-cH-prime}). A run $\rho$ of $\RegHA{\cH'}$ is
called a \emph{type-0 run} iff there are $\rho_0,\rho_1,\ldots,
\rho_k$ s.t. $\rho=\rho_0\cdot\rho_1\cdots\rho_k$, and for all $0\leq
i\leq k$: $\duration{\rho_i}<\frac{1}{\rmax}$. Then, each $\rho_i$
making up the type-0 run is called a \emph{type-1 run}.

\paragraph{Type-2 runs} Type-1 runs are further split into type-2 runs
as follows. Let $\rho= s_0, (t_1,e_1),s_1,\ldots, (t_{n}, e_{n}),s_n$
be a type-1 run of $\RegHA{\cH'}$, s.t. $\duration{\rho}\leq \tb$. Let
$S_\rho$ be the set of positions $0<i\leq n$ s.t:
$$
\exists x\in X:\left(
    \begin{array}{c}
      \ipart{\val_{i-1}(x)}\neq \ipart{\val_{i}(x)}\\
      \textrm{ or }\\
      \ipart{\val_{i-1}(x)}>0\textrm{ and }0=\fpart{\val_{i-1}(x)}<\fpart{\val_i(x)}
    \end{array}\right)
$$
where $\ipart{x}$ and $\fpart{x}$ denote respectively the integral and
fractional parts of $x$.  Roughly speaking, each transition
$(t_i,e_i)$ with $i\in S_\rho$ corresponds to the fact that a variable
changes its region, except in the case where the variable moves from
$\zplus $ to $(0,1)$: such transitions are not recorded in
$S_\rho$. Since $\rho$ is a type-1 run, its duration is at most
$\frac{1}{\rmax}$. Hence, each variable can cross an integer value at
most once along $\rho$, because all rates are positive. Thus, the size
of $S_\rho$ can be bounded, by a value \emph{that does not depend on
  $\len{\rho}$}:
\begin{lemma}\label{lemma:bound-number-of-times-a-clock-changes-integral}
  Let $\rho$ be a type-1 run. Then $|S_{\rho}|\leq 3\times|X|$.
\end{lemma}
\begin{proof}
  As the duration of a type-1 run is $<\frac{1}{\rmax}$, each variable
  can, in the worst case, follow a trajectory that will be split into
  4 parts. This happens when it starts in $(b,b+1)$, moves to
  $[b+1,b+1]$, then $(b+1,b+2)$, then gets reset and stays in $[0,1)$.
\end{proof}
Remark that if we had recorded in $S_\rho$ the indices of the
transitions from $(\ell,\val)$ to $(\ell',\val')$ s.t. $\val(x)=0$ and
$\val(x)\in (0,1)$ for some variable $x$,
Lemma~\ref{lemma:bound-number-of-times-a-clock-changes-integral} would
not hold, and we could not bound the size of $S_\rho$ by a value
independent from $\len{\rho}$. Indeed, in any time interval, the
density of time allows a variable to be reset and to reach a strictly
positive value an arbitrary number of times.

Let us now split a type-1 run $\rho$ according to $S_\rho$. Assume
$\rho=s_0,(t_1,e_1), s_1,\ldots,\break (t_n,e_n),s_n$, and that
$S_\rho=\{p(1),\ldots, p(k)\}$, with $p(1)\leq p(2)\leq\cdots\leq
p(k)$. Then, we let $\rho_0, \rho_1,\ldots,\rho_k$ be the runs s.t.:
\begin{align}
  \rho&=\rho_0\cdot s_{p(1)-1},(t_{p(1)}, e_{p(1)}),s_{p(1)}\cdot\rho_1 \cdot s_{p(2)-1},(t_{p(2)}, e_{p(2)}), \nonumber\\
  &\phantom{=}
  s_{p(2)},\ldots, s_{p(k)-1},(t_{p(k)}, e_{p(k)}),s_{p(k)}\cdot
  \rho_k\label{eq:decomposition-of-type-2}
\end{align}
Each $\rho_i$ is called a \emph{type-2 run}, and can be empty. The
next lemma summarises the properties of this construction:

\begin{lemma}\label{lem:type1 into type2}
  Let $\rho$ be a type-1 run of $\RegHA{\cH'}$ with
  $\duration{\rho}\leq \tb$. Then, $\rho$ is split into: $\rho_0\cdot
  \rho_1'\cdot\rho_1\cdot\rho_2'\cdot\rho_2\cdots\rho_k'\cdot\rho_k$
  where each $\rho_i$ is a type-2 run; $k\leq 3\times |X|$;
  $\len{\rho_i'}=1$ for all $1\leq i\leq k$; and for all $1\leq i\leq
  k$: $\rho_i=(\ell_0,\val_0),(t_1,e_1),\ldots,
  (t_n,e_n),(\ell_n,\val_n)$ implies that, for all $x\in X$:
  \begin{itemize}\item 
    \emph{either} there is $a\in\mathbb{N}^{>0}$ s.t. for all $0\leq
    j\leq n$: $\val_j(x)=a$ and $x$ is not reset along $\rho_i$;
  \item \emph{or} for all $0\leq j\leq n$: $\val_j(x)\in (a,a+1)$ with
    $a\in\mathbb{N}^{>0}$ and $x$ is not reset along $\rho_i$;
  \item \emph{or} for all $0\leq j\leq n$: $\val_j(x)\in [0,1)$.
  \end{itemize}
\end{lemma}
Remark that in the last case (i.e., $x$ is in $[0,1)$ along a type-2
run), the number of resets cannot be bounded \textit{a
  priori}. 
For the sake of clarity, we summarise the construction so far by the
following lemma:
\begin{lemma}\label{lem:split-of-type-0-runs}
  Each type-0 run of $\RegHA{\cH'}$ can be decomposed into $k$ type-2
  runs with $k\leq 3\times (T\times\rmax+1)\times |X|$.
\end{lemma}

\paragraph{Type-3 runs} Finally, we obtain type-3 runs by splitting
type-2 runs according to the first and last resets (if they exist) of
each clock. Formally, let $s_0,(t_1,e_1),s_1,\ldots,\break (t_n,e_n),s_n$ be
a type-2 run. Assume $Y_i$ is the reset set of $e_i$, for all $1\leq
i\leq n$. We let $FR_\rho=\{i\mid x\in Y_i\textrm{ and }\forall 0\leq
j<i: x\not \in Y_j\}$ and $LR_\rho=\{i\mid x\in Y_i\textrm{ and }
\forall i<j\leq n: x\not\in Y_j\}$ be respectively the set of edge
indices where a variable is reset for the first (last) in $\rho$. Let
$R_\rho=FR_\rho\cup LR_\rho$ and assume $R_\rho=\{p(1), p(2),\ldots,
p(k)\}$ with $p(1)\leq p(2)\leq\cdots\leq p(k)$. Then, we let $\rho_0,
\rho_1,\ldots, \rho_k$ be the \emph{type 3 runs} making up $\rho$ s.t.
$\rho=\rho_0\cdot s_{p(1)-1},(t_{p(1)},e_{p(1)}),s_{p(1)}\cdot
\rho_1\cdots s_{p(k)-1},(t_{p(k)},e_{p(k)}),s_{p(k)}\cdot
\rho_k$. Remark that each type-2 is split into at most $2\times|X|+1$
type-3 runs (i.e., $k\leq 2\times |X|$).


\paragraph{Contraction operator} 
So far, we have defined a procedure that splits any time-bounded run
of $\RegHA{\cH}$ into a bounded number of type-3 runs. However, the
construction does not allow us to bound the length of type-3 runs,
because the density of time allows to perform an arbitrary number of
actions in every possible time delay. Let us now define a contraction
operator that turns type-3 runs into runs with the same effect but
whose lengths can be uniformly bounded (thanks to the properties of
type-3 runs established below).

Intuitively, the contraction operator works as follows. Let
$\rho=(\ell_0, \val_0),(t_1,
e_1),\break (\ell_1,\val_1),\ldots,(t_n,e_n),(\ell_n,\val_n)$ be a run, and
let $\pi$ be its timed path. We \emph{contract} $\pi$ by looking for a
pair of positions $i<j$ s.t. $\ell_i=\ell_j$ (i.e., $\pi[{i+1}:j]$
forms a loop) \emph{and} s.t. all locations
$\ell_{i+1},\ell_{i+2},\ldots,\ell_{j}$ occur in the prefix
$\pi[1:i]$. This situation is depicted in
Fig.~\ref{fig:illustrate-contraction} (top). Then, the contraction
consists, roughly speaking, in \emph{deleting} the portion
$\pi[i+1:j]$ from $\pi$, and in \emph{reporting} the delays
$t_{i+1}$,\ldots, $t_{j-1}$ to the other occurrences of
$\ell_i,\ldots, \ell_{j-1}$ in $\pi$ (that exist by hypothesis), see
Fig.~\ref{fig:illustrate-contraction} (bottom). Clearly, in general,
the resulting timed path might not yield a run as some guards could
fail because of the additional delays. Yet, we prove (see
Proposition~\ref{prop:bound-on-length-of-contr-type-2}) that, when
carefully applied to type-2 runs, the contraction operator produces a
\emph{genuine run} with a bounded length, and that reaches the same
state as the original run. Remark that the proof of soundness of the
contraction operator relies on the fact that we have encoded the
regions of the variable valuations in the locations. This information
will be particularly critical when a variable is in $[0,1)$ and reset.

The contraction operator is first defined on timed paths (we will
later lift it to type-2 runs).  Let us consider a timed path $\pi =
(t_1,e_1),(t_2,e_2),\ldots,(t_n,e_n)$. Let $\ell_0=\src{e_1}$, and,
for all $1\leq i\leq n$: $\ell_i=\trg{e_i}$. Assume there are $0\leq
i< j < n$ and a function $h:\{i+1,\ldots,
j-1\}\mapsto\{0,\ldots,i-1\}$ s.t. $(i)$ $\ell_i=\ell_j$ and $(ii)$
for all $i< p< j$: $\ell_p= \ell_{h(p)}$.  Then, we let
$\contract{\pi}=\ell_0',(t_1',e_1'),\ldots,\ell_m'$ where:
\begin{enumerate}
\item $m=n-(j-i)$.
\item for all $0\leq p\leq i$: $\ell_p'=\ell_p$.
\item for all $1\leq p\leq i$: $e_p'=e_p$ and
  $t_p'=t_p+\sum_{k\in h^{-1}(p-1)}t_{k+1}$.
\item $e_{i+1}'=e_{j+1}$ and $t_{i+1}'=t_{i+1}+t_{j+1}$
\item for all $i+1<p\leq m$: $\ell_p'=\ell_{p+j-i}$ and
  $(t_p',e_p')=(t_{p+j-i},e_{p+j-i})$.
\end{enumerate}

Then, given a timed path $\pi$, we let $\contracti{\pi}{0}=\pi$,
$\contracti{\pi}{i}=\contract{\contracti{\pi}{i-1}}$ for any $i\geq
1$, and $\contractstar{\pi}=\contracti{\pi}{n}$ where $n$ is the least
value such that $\contracti{\pi}{n}=\contracti{\pi}{n+1}$. Clearly,
since $\pi$ is finite, and since $\len{\contract{\pi}}<\len{\pi}$ or
$\contract{\pi}=\pi$ for any~$\pi$, $\contractstar{\pi}$ always
exists. Moreover, we can always bound the length of
$\contractstar{\pi}$ by a value \emph{that does not depend on
  $\len{\pi}$}.

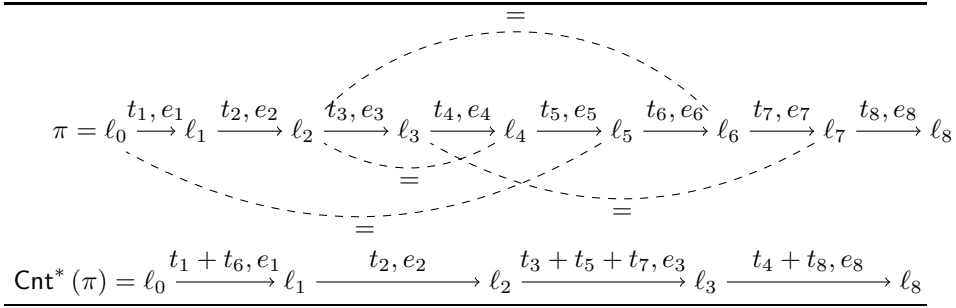
\begin{figure*}
  \centering
  \hrule
  \begin{tikzpicture}[node distance = 1.4 cm]
    \node (n0) at (0,0) {$\pi =\ell_0$} ;
    \foreach \x/\y in {1/0, 2/1, 3/2, 4/3, 5/4, 6/5, 7/6, 8/7} {
      \node (n\x) [right of=n\y] {$\ell_\x$} ;
    }
    
    \foreach \x/\y in {n5/n0, n4/n2, n7/n3} {
      \path[-, dashed, bend left]  (\x) edge node [below] {$=$} (\y) ;
    }
    \path[-, dashed, bend right=45]  (n6) edge node [above] {$=$} (n2) ;

    \foreach \x/\y in {0/1,1/2, 2/3, 3/4, 4/5, 5/6, 6/7, 7/8} {%
      \path[->] (n\x) edge node [above] {$t_{\y}, e_{\y}$} (n\y) ;
    }
    
    \node (m0) [below of=n0, node distance=2cm] {$\contractstar{\pi}=\ell_0$} ;
    
    \foreach \x/\y in {1/0,2/1,3/2} {%
      \node (m\x) [right of=m\y, node distance=2.7cm] {$\ell_{\x}$} ; 
    }
    \node (m4) [right of=m3, node distance=2.7cm] {$\ell_8$} ;

    \path[->] (m0) edge node [above] {$t_1+t_6, e_1$} (m1) 
    (m1) edge node [above] {$t_2,e_2$} (m2)
    (m2) edge node [above] {$t_3+t_5+t_7,e_3$} (m3)
    (m3) edge node [above] {$t_4+t_8,e_8$} (m4)
    ;
  \end{tikzpicture}
  \hrule
  \caption{Illustrating the contraction operator. Here, $i=3$, $j=7$,
    $h(4)=2$, $h(5)=0$ and $h(6)=2$.}
  \label{fig:illustrate-contraction}
\end{figure*}

\begin{lemma}\label{lem:length-contractstar}
  For all timed path $\pi$: $\len{\contractstar{\pi}}\leq
  |\loc|^2+1$.
\end{lemma}
\begin{proof}
  Assume $\pi'=\contractstar{\pi}=(t_1,e_1),(t_2,e_2),\ldots,\break
  (t_n,e_n)$. Let $\ell_0=\src{e_1}$, and $\ell_i=\trg{e_i}$ for all
  $1\leq i\leq n$. Let $\loc'=\{L_0,\ldots,L_m\}\subseteq \loc$ be the
  set of locations that appear in $\pi'$.  For all $L_i\in \loc'$, let
  $k_i$ denote the least index s.t. $\ell_{k_i}=L_i$ (i.e., the first
  occurrence of $L_i$ in $\pi'$). Wlog, we assume that $k_0\leq
  k_1\leq \cdots\leq k_m$. Then, clearly, $k_0=0$. Observe that each
  portion of the form $\pi'[k_i: k_{i+1}-1]$ (with $0\leq i\leq m-1$)
  is of length at most $|\loc'|$. Otherwise, the contraction operation
  can be applied in this portion, as there must be two positions
  $k_i\leq \alpha<\beta\leq k_{i+1}-1$ s.t. $\ell_\alpha=\ell_\beta$,
  and all the locations occurring along $\pi'[\alpha:\beta-1]$ have
  occurred before, by definition of $k_i$ and $k_{i+1}$. By the same
  arguments, $\len{\pi'[k_m:n-1]}\leq |\loc'|$ (remark that, by
  definition of the contraction operator, the last location $\ell_n$
  will never be considered for contraction). As $\pi'[0:n-1]$ is made
  up of all those portions, and as there are $|\loc'|$ portions,
  $|\pi'|$ is bounded by $|\loc'|^2+1\leq |\loc|^2+1$.
\end{proof}

We can now lift the definition of the contraction operator to
\emph{runs} of type-2. Let $\rho$ be a type-2 run and let us consider
its (unique) decomposition into type-3 runs, as in
(\ref{eq:decomposition-of-type-2}), above. 
Then, we let $\contract{\rho}=\rof{\first{\rho}}{\pathname}$, where:
\begin{align*}
  \pathname &= \contractstar{\tpof{\rho_0}},(t_{p(1)}, e_{p(1)}),\contractstar{\tpof{\rho_1}},\\
  &\phantom{=} (t_{p(2)}, e_{p(2)}),\ldots, (t_{p(k)},
  e_{p(k)}),\contractstar{\tpof{\rho_k}}
\end{align*}

By definition of $\contractstarname$, and by definition of $\contractname$
on type-2 runs, it is easy to see that:
\begin{lemma}\label{lem:effect}
  For all type-3 runs $\rho$:
  $\duration{\contractstar{\tpof{\rho}}}=\duration{\rho}$ and for all
  variables $x$:
  $\effect{\contractstar{\tpof{\rho}}}(x)=\effect{\tpof{\rho}}(x)$.
  Similarly, for all type-2 runs $\rho$:
  $\duration{\pathname}=\duration{\rho}$ and for all
  variables~$x$:\break $\effect{\pathname}(x)=\effect{\tpof{\rho}}(x)$.
\end{lemma}

Let us show that the contraction of type-2 runs is sound:
\begin{proposition}\label{prop:bound-on-length-of-contr-type-2}
  For all type-2 runs $\rho$, $\contract{\rho}\neq \bot$,
  $\first{\contract{\rho}}=\first{\rho}$ and
  $\last{\contract{\rho}}=\last{\rho}$.
\end{proposition}
\begin{proof}
  Let $\rho=(\ell_0,\val_0),(t_1,e_1),\ldots
  (t_n,e_n),(\ell_n,\val_n)$. Let $\pi$ denote $\tpof{\rho}$, and let
  $\pathname=(t_1',e_1'),\ldots,(t_k',e_k')$.  For all $1\leq i\leq
  k$, let $\ell_i'=\dest{e_i'}$ by $\ell_i'$; and let
  $\ell_0'=\src{e_1}=\ell_0$.

  First, observe that, by definition of the contraction operator,
  $\ell_n=\ell_k'$.  Let us show that $\contract{\rho}\neq
  \bot$. Assume that, for all $0\leq i\leq k$: $\ell_i'=(\ellbar_i,
  r_i)$ and let $\val_i'$ be the valuation s.t. for all $x$:
  \begin{align*} 
    \val_i'(x) &=
    \begin{cases}
      \val_{i-1}'(x)+\rate(\ell'_{i-1})(x)\times t'_i&\textrm{If }e_i'\textrm{ does not reset }x\\
      0&\textrm{Otherwise}
    \end{cases}
  \end{align*}
  Finally, let $\val_0'=\val_0$.  Remark that $\val_0'(x)\leq
  \val_1'(x)\leq\cdots\leq\val_k'(x)$ because rates are
  non-negative. Clearly, to show that $\rho'\neq\bot$, it is
  sufficient to show, for all $i$, that $\val_i'\models g_i'$ (where
  $g_i'$ is the guard of $e_i'$); and that both $\val_i$ and $\val_i'$
  satisfy\footnote{Remember that we consider \rhadp, so the invariants
    are convex.} $\invariants(\ell_i)$. For the sake of clarity, we
  prove that all the guards are satisfied; the arguments can be easily
  adapted to show that the invariants are satisfied too.

  First, consider a variable $x$ that is not reset along $\pi$ (hence
  along $\pathname$) and s.t. $\val_0(x)=\val_0'(x)> 0$. 
  By definition of type-2 runs, and since $x$ is not reset and not
  null initially, $\val_0(x)$, $\val_1(x)$,\ldots, $\val_n(x)$ all
  belong to the same interval $I$ which is either $(a-1,a)$ or $[a,a]$
  for some $a\geq 1$. Thus, in particular, $\val_0(x)=\val_0'(x)\in
  I$. Moreover, since $\effect{\pathname}(x)=\effect{\pi}(x)$
  (Lemma~\ref{lem:effect}), we have $\val'_k(x)=\val_n(x)\in I$ too.
  Hence, since $\val_0'(x)\leq \val_1'(x)\leq\cdots\leq\val_k'(x)$, we
  conclude that $\val_i'(x)\in I$ for all $0\leq i\leq k$. Since all
  the $\val_i(x)$ are also in $I$, since $\rho$ is a genuine run, and
  since all edges $e_i'$ in $\pi'$ are also present in $\pi$, we
  conclude that $\val\in I$ \emph{implies} $\val\models g_i'$, for all
  valuation $\val$ and all guards $g_i'$ of some edge $e_i'$ in
  $\pi$. Hence, $\val_i'\models g_i'$ for all $i$.

  Thus, we can, from now on, safely ignore all variables $x$ that are
  not reset along $\pi$ (hence along $\pathname$) and
  s.t. $\val_0(x)=\val_0'(x)> 0$, and focus on variables $x$ that are
  either reset along $\pi$ or s.t. $\val_0(x)=\val_0'(x)=0$. By
  definition of type-2 runs, in both cases, these variables take
  values in $[0,1)$ in each state along $\rho$. Hence, since
  $\rho$ is region consistent (Lemma~\ref{lemma:value-variables}), all
  locations in $\rho$ are of the form $(\ell, r)$ with
  $r(x)\in\{\zeq,\zplus,(0,1)\}$, and so are all locations in
  $\pathname$: for all $0\leq i\leq k$:
  $\ellbar_i\in\{\zeq,\zplus,(0,1)\}$. Let us denote, by $\rho'_j$ the
  value $\rof{\first{\rho}}{\pathname[1:j]}$ for all $m\geq 1$. We
  further denote by $\rho'_0$ the run of null length
  $(\ell_0',\val_0')$. Let us show that, for all $0\leq j\leq k$,
  $\rho'_j\neq \bot$, by induction on $j$.

  The base case is $j=0$ and is trivial since
  $(\ell_0,\val_0)=(\ell_0',\val_0')$. For the inductive case, we
  assume that $\rho'_{m-1}\neq \bot$ (for some $m\geq 1$) and ends in
  $((\ellbar, r), \val)$, and we show that we can extend it by firing
  $(t_{m}', e_m')$ (i.e., that $\rho'_{m}\neq \bot$). Observe that, by
  definition of $\contractname$, the edge $e_m'$ occurs in $\pathname$
  because it was already present in $\pi$ (say, at position $\alpha$,
  hence $e_\alpha=e_m'$ and $(\ellbar,r)=\ell_{\alpha-1}$). Moreover,
  still by definition of $\contractname$, the delay $t_m'$ is equal to
  $t_\alpha+\sum_{i=1}^\beta t_{p(i)}$, where for all $1\leq
  i\leq\beta$: $\src{e_{p(i)}}=(\ellbar,r)$.  We consider three cases:

  \begin{enumerate}
  \item Either $r(x)=\zeq$. In this case, since $\rho$ and
    $\rho'_{m-1}$ are region consistent
    (Lemma~\ref{lemma:value-variables}), and since the region $r(x)$
    is $\zeq$ (and not $\zplus$), we know that $\val_{\alpha-1}(x)=0$
    ($x$ is null when entering $(\ellbar, r)$ at position $\alpha-1$
    in $\rho$), that $\val(x)=0$ ($x$ is null at the end of
    $\rho'_{m-1}$), and that
    $\val_{\alpha-1}(x)+t_\alpha\times\rate(\ellbar,r)(x)=
    t_\alpha\times\rate(\ellbar,r)(x)=0$ ($x$ is null when leaving
    $(\ellbar,r)$ at position $\alpha-1$ in $\rho$). This means, in
    particular that it is sufficient for $x$ to be null to satisfy the
    guard of $e_m'=e_\alpha$. Moreover, for all $1\leq i\leq\beta$:
    $\val_{p(i-1)}(x)=0=t_{p(i)}\times\rate(\ellbar,r)(x)$ ($x$ is
    null when entering and leaving the locations at all positions
    $p(i)$ that have yielded the contraction in $\rho$). Thus, the
    value that $x$ takes after letting $t_m'$ t.u. elapse the last
    state or $\rho_{m-1}'$ is
    $\val'(x)=\val(x)+t_m'\times\rate(\ellbar,r)(x)=(t_\alpha+\sum_{i=1}^\beta
    t_{p(i)})\times\rate(\ellbar,r)(x)=0$. Hence $\val'(x)$ satisfies
    the guard of $e_m'$, and we can extend $\rho'_{m-1}$ by $(t_{m}',
    e_m')$.
  \item Or $r(x)=\zplus$. In this case, we know that
    $\val_{\alpha-1}(x)=\val(x)=0$, that
    $t_\alpha\times\rate(\ellbar,r)(x)>0$, and that for all $1\leq
    i\leq \beta$: $\val_{p(i-1)}(x)=0$ and
    $t_{p(i)}\times\rate(\ellbar,r)(x)>0$. Moreover, since
    $\duration{\rho}<\frac{1}{\rmax}$, we can precise this information
    and conclude that $t_\alpha\times\rate(\ellbar,r)(x)\in (0,1)$ and
    that for all $1\leq i\leq \beta$:
    $t_{p(i)}\times\rate(\ellbar,r)(x)\in (0,1)$. Thus, it is
    sufficient, to satisfy the constraints on $x$ in the guard of
    $e_m'$, that $x\in (0,1)$. Let us show that
    $\val'(x)=(t_\alpha+\sum_{i=1}^\beta
    t_{p(i)})\times\rate(\ellbar,r)(x)$ is in $(0,1)$ too.  We have
    $\val'(x)>0$ because $t_\alpha\times\rate(\ellbar,r)(x)>0$, as
    shown above. Moreover, $\val'(x)<1$ because
    $t_\alpha+\sum_{i=1}^\beta
    t_{p(i)}\leq\duration{\rho}<\frac{1}{\rmax}$, by def. of type-2
    runs. Thus, $\val'(x)$ satisfies the guard of $e_m'$ and we can
    extend $\rho'_{m-1}$ by $(t_{m}', e_m')$.
  \item Or $r(x)=(0,1)$. In this case, we can rely on the same
    arguments as above to show that $\val'(x)>0$, and that $\val'(x)$
    should be in $(0,1)$ to satisfy the guard of $e_m'$. The
    difference with the previous case is that $\val(x)\neq 0$ here,
    and we have to make sure that the additional delay accumulated on
    $(\ellbar,r)$ by the contraction operator does not increase $x$
    above $1$. This property holds because of the split of type-2 runs
    in type-3 runs, according to the first reset of each
    variable. More precisely, we consider two cases. Either
    $\ell_{\alpha-1}$ occurs, in $\rho$ in a type-3 run that takes
    place \emph{after} the first reset of $x$. In this case,
    $\val'(x)=\val(x)(t_\alpha+\sum_{i=1}^\beta
    t_{p(i)})\times\rate(\ellbar,r)(x)<1$, because all the $t_{p(i)}$
    also occur $\pi[\alpha:n]$ (i.e., after the first reset of $x$),
    and $\duration{\pi[\alpha:n]}<\frac{1}{\rmax}$. Or
    $\ell_{\alpha-1}$ occurs, in $\rho$ in a type-3 run that takes
    place \emph{before} the first reset of $x$. In this case,
    $\val'(x)=\val(x)(t_\alpha+\sum_{i=1}^\beta
    t_{p(i)})\times\rate(\ellbar,r)(x)\geq 1$ implies that, in $\rho$:
    $\val_{p(i)}\geq 1$, which contradicts the definition of type-2
    runs. Hence, $\val'(x)\in (0,1)$ and we can extend $\rho'_{m-1}$
    by $(t_{m}', e_m')$.
  \end{enumerate}
  
  Let us conclude the proof by showing that $\val_k'=\val_n$.  We
  consider three cases. First, $x$ is a variable that is not reset
  along $\rho$. Since $\effect{\contractstar{\pi}}(x)=\effect{\pi}(x)$
  (Lemma~\ref{lem:effect}), and since $\val_0=\val_0'$, we conclude
  that $\val_k'(x)=\val_n(x)$. Second, $x$ is a variable that is reset
  along $\rho$. Since the duration of a type-2 is at most
  $\frac{1}{\rmax}$, $\val_n(x)\in [0,1)$. Thus, we consider two
  further cases. \emph{Either $\val_n(x)=0$}. Since $\rho$ is
  region-consistent (Lemma~\ref{lemma:value-variables}), $\ell_n$ is
  of the form $(\ellbar, r)$ with $r(x)\in \{\zplus,\zeq\}$. However,
  $\ell_n=\ell_k'$, and since $\contract{\rho}$ is a run and hence
  region-consistent, we conclude that $\val_k'(x)=0$ too. \emph{Or
    $\val_n(x)\in (0,1)$}. In this case, it is easy to observe that
  $\val_n(x)$ depends only on the portion of $\rho$ that occurs after
  the last reset of $x$, i.e., $\val_n(x)=\effect{\pi[i+1:n]}(x)$,
  where $i$ is the largest position in $\rho$ s.t. $e_i$ resets
  $x$. By definition of the contraction operator on type 2 runs, $e_i$
  occurs at some position $\alpha$ of $\pathname$,
  i.e. $e_i=e'_\alpha$ and $e'_\alpha$ is the last edge of $\pathname$
  to reset $x$. Thus,
  $\val_k'(x)=\effect{\pathname[\alpha+1:k]}(x)$. However, by
  Lemma~\ref{lem:effect}, and by definition of the contraction of type
  2 runs:
  $\effect{\pathname[\alpha+1:k]}(x)=\effect{\pi[i+1:n]}(x)$. Hence,
  $\val_n(x)=\val_k'(x)$.
\end{proof}

Then, observe that, by the above definition, and by
Lemma~\ref{lem:length-contractstar}, we can bound the length of
$\contract{\rho}$ for type-2 runs $\rho$:
\begin{lemma}\label{lem:length-contractstar-type-2}
  For all type-2 runs: $\len{\contract{\rho}}\leq 8\times
  |\loc|^2\times |X|$.
\end{lemma}
\begin{proof}
  By definition of type-2 runs, and by
  Lemma~\ref{lem:length-contractstar}, $|\contract{\rho}|$ is at most
  $(2\times |X|+1)\times (|\loc|^2+1)+2\times |X|= 2\times
  (|X|+1)\times (|\loc|^2+1)$. However, wlog, $|\loc|\geq 1$ and
  $|X|\geq 1$. Hence $|X|+1\leq 2\times |X|$, $|\loc|^2+1\leq 2\times
  |\loc|^2$. Hence the lemma.
\end{proof}

\begin{sidewaysfigure}
  \centering
  \hrule
  \begin{tikzpicture}[node distance=1.5cm]
    \node[label=$\dot{x}>0$] (n1) at (0,0) {$\rho_1 =\big((\ell_1,[0,0] ), 0\big)$} ;
    \node [left of =n1, node distance= 3cm] {\textbf{Before contraction}:} ;
    \node (n2) [right=of n1] {$\big((\ell_2,[0,0] ), 0\big)$} ;
    \node[label=$\dot{x}>0$] (n3) [right=of n2] {$\big((\ell_1,[0,0] ), 0\big)$} ;
    \node (n4) [right=of n3] {$\big((\ell_2,[0,0] ), 0\big)$} ;
  
    \path[->] (n1) edge node (e1) [above] {$t_1=0$} (n2)
    (n2) edge node (e2) [below] {$x:=0$} (n3)
    (n3) edge node [above] {$t_2>0$} node (e3) [below] {$x:=0$} (n4) ;
   
    \node[text width=5cm] (rem1) [below left of=e1] {When we cross this
      edge, $x$ is null.} ;

    \path[->, dashed] (rem1) edge [bend right] (e1) ;

    \node[text width=6cm] (rem2) [below right of=e3, node distance=1.2cm] {When we cross this
      edge, $x$ is not null.} ;
    
    \path[->, dashed] (rem2) edge [bend left] (e3) ;

    \node (rem5) [above right of=n3] {\ldots and this location would
      be $(\ell_1,\zplus)$} ;

    \node (rem6) [above right of=n1] {With our definition, this
      location would be $(\ell_1,\zeq)$\ldots} ;

    \path[->, dashed] (rem5) edge [bend left] (n3) 
    (rem6) edge [bend left] (n1);

    \node[label=$\dot{x}>0$] (n1bis) at (0,-2) {$\rho_2 = \big((\ell_1,[0,0] ), 0\big)$} ;
    \node [left of= n1bis, node distance=3cm] {\textbf{After contraction}:};
    
    \node (n2bis) [right=of n1bis] {$\big((\ell_2,[0,0] ), \boldsymbol{t_1+t_2}\big)$} ;
    
    \path[->] (n1bis) edge node (e1bis) [above] {$t_1+t_2>0$} (n2bis) ;

    \node[text width=6cm] (rem3) [below left of=e1bis] {When we cross this
      edge, $x$ is not null.} ;

    \node[text width=5cm] (rem4) [node distance=5cm, right of=n2bis] {We reach a
      state where $x$ is not null anymore !} ;

    \path[->, dashed] (rem3) edge [bend right] (e1bis)
    (rem4) edge  (n2bis) ;
  \end{tikzpicture}
  \hrule
  \caption{An example that shows why the contraction operator fails if
    we use $[0,0]$ to characterise the variables that are null.}
  \label{fig:0-0-vs-zplus}
\end{sidewaysfigure}
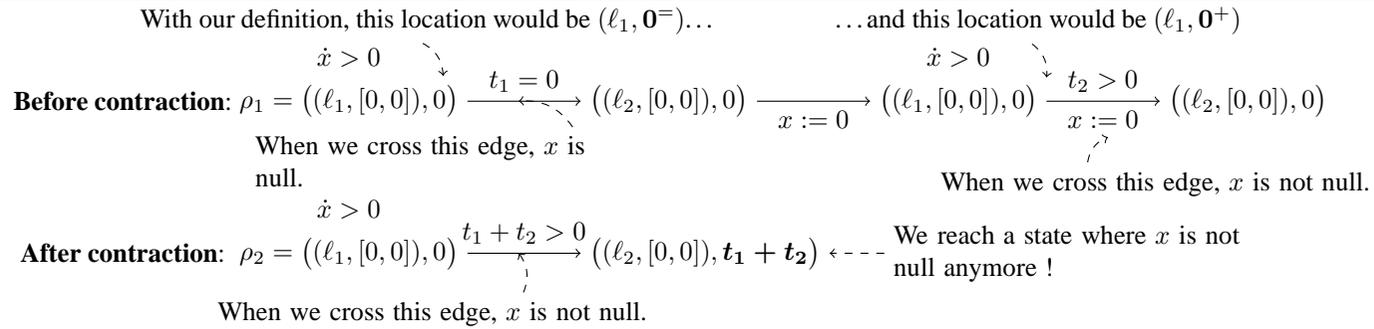

We can now explain more intuitively why we need two different regions
($\zeq $ and $\zplus$) for variables that are null, and cannot use
$[0,0]$ instead. Consider the example given in
Fig.~\ref{fig:0-0-vs-zplus}. Run $\rho_1$ depicts a run of an
automaton with a single variable $x$, where we have used only region
$[0,0]$ in the construction of $\RegHA{\cH'}$. In this run, $x$ is
null in all four states. The two locations of $\RegHA{\cH'}$ that are
met are $(\ell_1, [0,0] )$ and $(\ell_2, [0,0])$ (and in both
locations, the rate of $x$ is strictly positive). Hence, the
contraction operator `merges' the two occurrences of both locations,
an produces $\rho_2$. However, $\rho_2$ fails to satisfy
Proposition~\ref{prop:bound-on-length-of-contr-type-2}, as $x$ is null
in the last state of $\rho_1$ but not in the last state of
$\rho_2$. This comes from the fact that region $[0,0]$ does not allow
to distinguish between locations that are left with a strictly
positive delay or a null delay. With our definition of $\RegHA{\cH'}$,
however, the first state of the run is $\big((\ell_1,\zeq ),0\big)$,
as $x$ is null when crossing the first edge, but the \emph{third}
state is $\big((\ell_1,\zplus),0\big)$, as $x$ is \emph{not null} when
crossing the last edge, which avoids the problem illustrated in
Fig.~\ref{fig:0-0-vs-zplus}.

Thus, summing up the properties of the contraction operator, and the
splitting procedure we obtain, as a corollary of
Proposition~\ref{prop:bound-on-length-of-contr-type-2} and
Lemma~\ref{lem:split-of-type-0-runs}:
\begin{corollary}\label{cor:time-bounded-equals-len-bounded-on-type-0}
  Let $s$ and $s'$ be two states of $\RegHA{\cH'}$. Then,
  $\RegHA{\cH'}$ admits a $\tb$-time-bounded type-0 run $\rho$ with
  $\first{\rho}=s$ and $\last{\rho}=s'$ iff it admits a $\tb$-time
  bounded type-0 run $\rho'$ with $\first{\rho'}=s$, $\last{\rho'}=s'$
  and $\len{\rho'}\leq 48\times \tb\times\rmax\times |\loc'|^2\times
  |X|^2$, where $X$, $\loc'$ and $\rmax$ are resp. the set of
  variable, set of locations and maximal rate of $\RegHA{\cH'}$.
\end{corollary}

Finally, for all \shadp $\cH=(X, \loc,  \edges, \rate,
\invariants, \init)$ and all time bound $\tb\in\mathbb{N}$, we let:
$$
\begin{array}{c}
  F(\cH, \tb) = \\
  24\times (\tb\times\rmax+1)\times |X|^2 \times |\loc|^2\times
  (2\times\cmax+1)^{2\times |X|}
\end{array}
$$

This value $F(\cH,\tb)$ is actually a bound on the length of the runs
we need to consider to decide $\tb$-time-bounded reachability:

\begin{theorem}\label{theo:time-bounded-imply-bound-on-length}
  Let $\cH$ be a \shadp, $\tb$ be a time bound and let $s_1$ and $s_2$
  be two states of $\cH$. Then $\cH$ admits a $\tb$-time-bounded
  run $\rho$ with $\first{\rho}=s_1$ and $\last{\rho}=s_2$ iff it
  admits a $\tb$-time-bounded run $\rho'$ with $\len{\rho'}\leq
  F(\cH,\tb)$, $\first{\rho'}=s_1$ and $\last{\rho'}=s_2$.
\end{theorem}
\begin{proof}
  The \emph{if} direction is trivial, let us prove the \emph{only if},
  by proving the contraposition, i.e., that if $\cH$ admits no
  $\tb$-time-bounded run of length at most $F(\cH,\tb)$ from $s_1$ to
  $s_2$, then it admits no $\tb$-time-bounded run from $s_1$ to
  $s_2$. By Lemma~\ref{lem:cH-cH-prime}, if $\cH$ admits no $\tb$-time
  bounded run of length at most $F(\cH,\tb)$ from
  $s_1=(\ell_1,\val_1)$ to $s_2=(\ell_2,\val_2)$, then, $\cH'$ admits
  no $\tb$-time-bounded run of length at most $F(\cH,\tb)$ from $s_1$
  to $s_2$. Then, by Lemma~\ref{lemma:reg-ha-preserves-reach}, then,
  for all pair of regions $r_1$, $r_2$: $\RegHA{\cH}$ admits no type-0
  $\tb$-time-bounded run of length at most $F(\cH,\tb)$ from
  $s_1'=((\ell_1,r_1),\val_1)$ to $s_2'=((\ell_2,r_2),\val_2)$. By
  Corollary~\ref{cor:time-bounded-equals-len-bounded-on-type-0}, and
  by~(\ref{eq:number-of-locations-reg-H}), $\RegHA{\cH}$ admits no
  type-0 $\tb$-time-bounded run from $s_1'$ to $s_2'$, regardless of
  the length of the run. Hence, by
  Lemma~\ref{lemma:reg-ha-preserves-reach}, $\cH'$ admits no
  $\tb$-time-bounded run $\rho$ from $s_1$ to $s_2$, and neither does
  $\cH$, by Lemma~\ref{lem:cH-cH-prime} again.
\end{proof}
Remark that $F(\cH,\tb)=\mathcal{O}\left(\tb\times 2^{|\cH|}\right)$,
where $|\cH|$ is the number of bits necessary to encode $\cH$, using
standard encoding techniques and binary encoding for the
constants. Hence,
Theorem~\ref{theo:time-bounded-imply-bound-on-length} tells us that,
to decide $\tb$-time-bounded reachability, we only need to consider runs
whose length is singly exponential in the size of the instance $(\cH,
\tb)$.

Let us now briefly explain how we can adapt the previous construction
to cope with non-singular rates. Let us first notice that given $\cH$
a \rhadp, the construction of $\RegHA{\cH'}$ still makes perfect sense
and still satisfies Lemma~\ref{lemma:value-variables}. Then, we need
to adapt the definition of timed path. A timed path is now of the form
$(t_1,R_1,e_1)\cdots (t_n,R_n,e_n)$, where each
$R_i:X\mapsto\mathbb{R}$ gives the actual rate that was chosen for
each variable at the $i$-th continuous step. It is then
straightforward to extend the definitions of $\contractname$, ${\sf
  Effect}$ and ${\sf Contraction}$ to take those rates into account
and still keep the properties needed to prove
Theorem~\ref{theo:time-bounded-imply-bound-on-length}. More precisely,
the contraction of a set of transitions $(t_1, R_1, e_1),\ldots, (t_n,
R_n, e_n)$ yields a transition $(t, R, e)$ with $t=\sum_{i=1}^n t_i$
and, $R=\frac{\sum_{i=1}^n t_i\times R_i}{t}$. Note that we need to
rely on the convexity of the invariants and rates in an RHA to ensure
that this construction is correct.  Thus, we can extend
Theorem~\ref{theo:time-bounded-imply-bound-on-length} to the case of
RHA with positive rates (\rhadp):

\begin{corollary}\label{cor:theo-for-rha}
   Let $\cH$ be a \rhadp, $\tb$ be a time bound and let $s_1$ and $s_2$
  be two states of $\cH$. Then $\cH$ admits a $\tb$-time-bounded
  run $\rho$ with $\first{\rho}=s_1$ and $\last{\rho}=s_2$ iff it
  admits a $\tb$-time-bounded run $\rho'$ with $\len{\rho'}\leq
  F(\cH,\tb)$, $\first{\rho'}=s_1$ and $\last{\rho'}=s_2$.
\end{corollary}


%% file: algorithm.tex
\section{Time-bounded reachability is NEXPTIME-c}
\label{sec:time-bound-reach}

In this section, we establish the exact computational complexity of
the time-bounded reachability problem for \rhadp.

\begin{theorem}\label{theo:nexp-time-complete}
  The time-bounded reachability problem for \rhadp is complete for
  {\sc NExpTime}.
\end{theorem}

To prove this theorem, we exhibit an {\sc NexpTime} algorithm for
time-bounded reachability and we reduce this problem from the
reachability problem of exponential time Turing machine.


\paragraph{An {\sc NexpTime} algorithm}
Recall that an instance of the time-bounded reachability problem is of
the form $(\H, \ell, \tb)$, where $\H$ is an \rhadp, $\ell$ is a
location, and $\tb$ is a time bound (expressed in binary). We
establish membership to {\sc NexpTime} by giving a non-deterministic
algorithm that runs in exponential time in the size of $(\H, \ell,
\tb)$ in the worst case. The algorithm first \emph{guesses} a sequence
of edges ${\cal E}=e_0 e_1 \dots e_n$ of $\cH$ s.t.  $n+1 \leq
F(\cH,\tb)$ and $\trg{e_n}=\ell$. Then the algorithm builds from
${\cal E}$ a linear constraint $\Phi({\cal E})$ , that expresses all
the properties that must be satisfied by a run that follows the
sequence of edges in ${\cal E}$ (see~\cite{DBLP:conf/hybrid/JhaKWC07}
for a detailed explanation on how to build such a constraint). This
constraint uses $n+1$ copies of the variables in $X$ and $n+1$
variables $t_i$ to model the time elapsing between two consecutive
edges, and imposes that the valuations of the variables along the run
are consistent with the rates, guards and resets of $\cH$. Finally,
the algorithm checks whether $\Phi({\cal E})$ is satisfiable and
returns `yes' iff it is the case.

The number of computation steps necessary to build $\Phi({\cal E})$
is, in the worst case, exponential in $|\cH|$ and $\tb$. Moreover,
checking satisfiability of $\Phi({\cal E})$ can be done in polynomial
time (in the size of the constraint) using classical algorithms to
solve linear programs. Clearly this procedure is an {\sc NExpTime}
algorithm for solving the time-bounded reachability problem for \rhadp.


\paragraph{{\sc NexpTime}-hardness}
To establish the {\sc NExpTime}-hardness, we show how to reduce the
membership problem for non-deterministic exponential time Turing
machines to time-bounded reachability for \shadp.

A non-deterministic exponential time Turing machine (\nexptm) is a
tuple $M=(Q,\Sigma,\Gamma,\sharp,q_0,\delta,F,\xi)$ where $Q$ is the
(nonempty and finite) set of control states, $\Sigma$ is the (finite)
input alphabet, $\Gamma\supseteq \Sigma$ is the (finite) alphabet of
the tape, $\sharp\in \Gamma$ is the blank symbol, $q_0 \in Q$ is the
initial control state, $\delta \subseteq Q \times \Gamma \times \Gamma
\times \{L,R\} \times Q$ is the transition relation, $F \subseteq Q$
is the set of accepting states, and
$\xi=\mathcal{O}\left(2^{p(n)}\right)$ (for some polynomial $p$), is
an exponential function that bounds the execution time of the machine
on input $w$ by $\xi(|w|)$.

As usual, a state of $M$ is a triple $(q,w_1,w_2)$ where $q \in Q$ is
a control state, $w_1 \in \Gamma^{*}$ a word that represents the
content of the tape on the left of the reading head (this word is
empty when the head is on the leftmost cell of the tape), and $w_2 \in
\Gamma^{*}$ is the content of the tape on the right of the reading
head excluding the sequence of blank symbols ($\sharp$) at the end of
the tape, (in particular the first letter in $w_2$ is the content of
the cell below the reading head).

A transition of the Turing machine is a tuple of the form
$(q_1,\gamma_1,\linebreak\gamma_2,D,q_2)$ with the usual semantics: it
is enabled iff the current control state is $q_1$, the content of the
cell below the reading head is equal to $\gamma_1$, and the head
should not be above the left most cell when $D=L$. The execution of the
transition modifies the content of the tape below the reading head to
$\gamma_2$, moves the reading head one cell to the right if $D=R$, or
one cell to the left if $D=L$, and finally, changes the control state
to $q_2$. We write $(q,w_1,w_2) \triangleright (q',w'_1,w'_2)$ if
there exists a transition in $\delta$ from state $(q,w_1,w_2)$ to
state $(q',w'_1,w'_2)$.

An (exponentially bounded) execution of $M$ on input $w$ is a finite
sequence of states $c_0 c_1 \dots c_n$ such that: $(i)$ $n \leq
\xi(|w|)$ (the execution is exponentially bounded); $(ii)$
$c_0=(q_0,\epsilon,w \cdot \sharp^{\xi(|w|)-|w|})$, (the initial
control state is $q_0$ and the tape contains $w$ followed by the
adequate number of blank symbols); and $(iii)$ for all $0 \leq i < n$
$c_i \triangleright c_{i+1}$, (the transition relation is enforced).
The execution is \emph{accepting} iff $c_n=(q,w_1,w_2)$ with $q \in
F$. W.l.o.g., we make the assumption that $\Sigma=\{0,1\}$,
$\Gamma=\{0,1,\sharp\}$, and transitions only write letters in
$\Sigma$. This ensures that in all reachable states $(q,w_1,w_2)$ we
have that $w_1,w_2 \in \{ 0,1 \}^{*}$.
 
The {\em membership problem} for an \nexptm $M$ and a word $w$ asks
whether there exists an accepting execution of the Turing Machine $M$
that uses at most $\xi(|w|)$ steps.

Let us show how we can encode all executions of $M$ into the
executions of an \shadp $\cH_M$. We encode the words $w_1$ and $w_2$
as pairs of rational values $(l_1,c_1)$ and $(l_2,c_2)$ where
$l_i=\frac{1}{2^{|w_i|}}$ encodes the length of the word $w_i$ by a
rational number in $[0,1]$, and $c_i$ encodes $w_i$ as follows. Assume
$w_1=\sigma_0 \sigma_1 \dots \sigma_n$. Then, we let
$c_1=\valL{w_1}=\sigma_n \cdot \frac{1}{2} + \sigma_{n-1} \cdot
\frac{1}{4}+ \dots + \sigma_0 \cdot \frac{1}{2^{n+1}}$. Intuitively,
$c_1$ is the value which is represented in binary by
$0\mathtt{.}\sigma_n\sigma_{n-1}\cdots \sigma_0$, i.e., $w_1$ is the
binary encoding of the fractional part of $c_1$ where the most
significant bit in the rightmost position.  For instance, if
$w_1=001010$ then $\valL{w_1}=0 \cdot \frac{1}{2} + 1 \cdot
\frac{1}{4} + 0 \cdot \frac{1}{8} + 1 \cdot \frac{1}{16}+0 \cdot
\frac{1}{32} + 0 \cdot \frac{1}{64} =0.3125$, and so $w_1$ is encoded
as the pair $(\frac{1}{64},0.3125)$. Remark that we need to remember
the actual length of the word $w_1$ because the function
$\valL{\cdot}$ ignores the leading $0$'s (for instance,
$\valL{001010}=\valL{1010}$). Symmetrically, if $w_2=\sigma_0 \sigma_1
\dots \sigma_n$, we let $c_2=\valR{w_2}=\sigma_0 \cdot \frac{1}{2} +
\sigma_{1} \cdot \frac{1}{4}+ \dots + \sigma_n \cdot
\frac{1}{2^{n+1}}$ (i.e., $\sigma_0$ is now the most significant
bit). 
Then, a state $(q, w_1, w_2)$ of the TM is encoded as follows: the
control state $q$ is remembered in the locations of the automaton, and
the words $w_1$, $w_2$ are stored, using the encoding described above
using four variables to store the values $(l_1, c_1)$ and $(l_2,
c_2)$.

With this encoding in mind, let us list the operations that we must be
able to perform to simulate the transitions of the TM. Assume
$w_1=w^1_0w^1_2\cdots w^1_n$ and $w_2=w^2_0w^2_2\cdots w^2_k$. We
first describe the operations that are necessary to \emph{read the
  tape}:
\begin{itemize} 
\item \emph{Read the letter under the head}. Following our encoding,
  we need to test the value of the bit $w^2_0$. Clearly, $w^2_0=1$ iff
  $l_2 \leq 1/2$, and $c_2 \geq \frac{1}{2}$; $w^2_0=0$ iff $l_2 \leq
  1/2$, and $c_2 < \frac{1}{2}$ and $w^2_0=\sharp$ iff $l_2 = 1$
  (which corresponds to $w_2=\epsilon$).
%
%
\item \emph{Test whether the head is in the leftmost cell of the
    tape}. This happens if and only if $w_1 = \epsilon$, and so if
  and only if $l_1 = 1$.
\item \emph{Read the letter at the left of the head} (assuming that
  $w_1 \neq \epsilon$). Following our encoding, this amounts to
  testing the value of the bit $w^1_n$. Clearly, $w^1_n=1$ iff $c_1
  \geq \frac{1}{2}$ and $w^1_n=0$ iff $c_1 < \frac{1}{2}$.
\end{itemize}
Then, let us describe the operations that are necessary to update the
values on the tape. Clearly, they can be carried out by appending and
removing $0$ or $1$'s to the right of $w_1$ or to the left of
$w_2$. Let us describe how we update $c_1$ and $l_1$ to simulate these
operations on $w_1$ (the operations on $w_2$ can be deduced from this
description). We denote by $c_1'$ (resp. $l_1'$) the value of $c_1$
($l_1$) after the simulation of the TM transition.
\begin{itemize}
\item To append a $1$ to the right of $w_1$, we let
  $l_1'=\frac{1}{2}\times l_1$. We let $c'_1=\frac{1}{2}$ if
  $l_1=1$ (i.e. $w_1$ was empty) and $c'_1=\frac{1}{2}
  \times c_1 + \frac{1}{2}$.
\item To append a $0$ to the right of $w_1$, we let
  $l'_1=\frac{1}{2}\times l_1$ and $c'_1=\frac{1}{2} \times c_1$.
\item To delete a $0$ from the rightmost position of $w_1$, we $l_1'=2
  \times l_1$, $c'_1=2 \times c_1$.
\item To delete a $1$ from the rightmost position of $w_1$, $l_1'=2
  \times l_1$, and $c'_1=(c_1-\frac{1}{2}) \times 2$.
\end{itemize}
In addition, remark that we can flip the leftmost bit of $w_2$ by
adding or subtracting $1/2$ from $c_2$ (this is necessary when
updating the value under the head).

Thus, the operations that we need to be able to perform on $c_1$,
$l_1$, $c_2$ and $l_2$ are: to multiply by $2$, divide by $2$,
increase by $\frac{1}{2}$ and decrease by $\frac{1}{2}$, while keeping
untouched the value of all the other variables. Fig.~\ref{fig:multi2}
exhibits four gadgets to perform these operations. Remark that these
gadgets can be constructed in polynomial time, execute in exactly 1
time unit time and that all the rates in the gadgets are singular.

We claim that all transitions of $M$ can be simulated by combining the
gadgets in Fig.~\ref{fig:subadd12} and the tests described above.  As
an example, consider the transition:\break $(q_1,1,0,L,q_2)$. It is
simulated in our encoding as follows. First, we check that the reading
head is not at the leftmost position of the tape by checking that $l_1
< 1$. Second, we check that the value below the reading head is equal
to $1$ by testing that $l_2 < 1$ and $c_2 \geq \frac{1}{2}$. Third, we
change the value below the reading head from $1$ to $0$ by subtracting
$\frac{1}{2}$ from $c_2$ using an instance of gadget $(ii)$ in
Fig.~\ref{fig:subadd12}. And finally, we move the head one cell to the
left. This is performed by testing the bit on the left of the head,
deleting it from $w_1$ and appending it to the left of $w_2$, by the
operations described above. All other transitions can be simulated
similarly. Remark that, to simulate one TM transition, we need to
perform several tests (that carry out in 0 t.u.) and to: $(i)$ update
the bit under the reading head, which takes $1$ t.u. with our gadgets;
$(ii)$ remove one bit from the right of $w_1$ (resp. left of $w_2$),
which takes at most 3 t.u. and $(iii)$ append this bit to the left of
$w_2$ (right of $w_1$), which takes at most 3 t.u. We conclude that
each TM transition can be simulate in at most $7$ time units.


Thus $M$ has an accepting execution on word $w$ (of length at most
$\xi(|w|)$ iff $\cH_M$ has an execution of duration at most $\tb=7
\cdot \xi(|w|)$ that reaches a location encoding an accepting control
state of $M$. This sets the reduction.

\begin{sidewaysfigure}
  \centering
  \hrule
 \begin{tikzpicture}[node distance=1.4cm]
  \tikzset{
    state/.style={
           rectangle,
           rounded corners,
           inner sep=2pt,
           text centered,
           draw
           },
         }
    \node (firsta) at (0,0) {} ;
    \node [above of= firsta, node distance=1cm] {$(i)$} ;
    \node[state, label=$x\leq 1$] (state1a) [right=of firsta]  {%
      $
      \begin{array}{rcl}
        \dot{x} &= &1\\
        \dot{z} &= &1\\
      \end{array}
      $
    } ;
    
    \node[state, label=$z\leq 1$] (state2a) [right =of state1a] {%
      $
      \begin{array}{rcl}
        \dot{x} &= &2\\
        \dot{z} &= &1\\
      \end{array}
      $
    } ;
    \node (lasta)  [right= of state2a] {};
    
    \path[->] (firsta) edge node  [above] {$z:=0$} (state1a)  
    (state1a) edge node (rem1dest) [above, text width=1.2cm] {$x=1$ $x:=0$} (state2a) 
    (state2a) edge node (rem2dest) [above] {$z=1$} (lasta);
    
    \node (rem1a) [ right  = of lasta, text width=3cm, node distance=1cm] {
      When crossing this edge, $z=1-x_0$.
    } ;


    \path[->,dashed] (rem1a) edge [bend right] (rem1dest) ;


    \node (first) [below  of=firsta, node distance=2.5cm] {} ;
     \node [above of= first, node distance=1cm] {$(ii)$} ;
    \node[state, label=$z\leq 1$] (state1) [right=of first]  {%
      $
      \begin{array}{rcl}
        \dot{x} &= &1\\
        \dot{z} &= &2\\
     \end{array}
      $
    } ;
    \node (last) [right=of state1] {} ;

    \path[->] (first) edge node (e1) [above] {$z:=0$} (state1)
    (state1) edge node (e2) [above] {$z=1$} (last) ;




    \node[state, label=$x\leq 1$] (state2) [right of=last, node distance=3cm]  {%
      $
      \begin{array}{rcl}
        \dot{x} &= &1\\
        \dot{z} &= &1\\
      \end{array}
      $
    } ;

    \node (first2) [left=of state2] {} ;
   \node [above of= first2, node distance=1cm] {$(iii)$} ;
    
    \node[state, label=$z\leq 1$] (state3) [right=of state2] {%
      $
      \begin{array}{rcl}
        \dot{x} &= &1\\
        \dot{z} &= &1\\
      \end{array}
      $
    } ;
    
    \node (last2) [right=of state3] {} ;

    \path[->] (first2) edge node (e3) [above] {$z:=1/2$} (state2)
    (state2) edge node (e4) [above, text width=1cm] {$x=1$ $x:=0$}
    (state3) 
    (state3) edge node (e5) [above] {$z=1$} (last2) ;
    
    \node (rem4) [above right of=last2, text width=3.5cm, node distance=1.5cm] {$z=1/2+(1-x_0)$ when crossing this edge} ;

    \path[->, dashed] (rem4) edge [bend right] (e4) ;
  \end{tikzpicture}
  \vspace*{10pt}
  \hrule
  \caption{Gadgets $(i)$ for multiplication by $2$, $(ii)$ adding
    $\frac{1}{2}$ and $(iii)$ subtracting $\frac{1}{2}$. The rates of
    the $y\not\in \{x,z\}$ is $0$. Gadget $(i)$ can be modified to
    divide by $2$, by swapping the rates of $x$ and $z$ in the second
    location. $x_0$ is the value of $x$ when entering the gadget.}
   \label{fig:subadd12}\label{fig:multi2}
\end{sidewaysfigure}
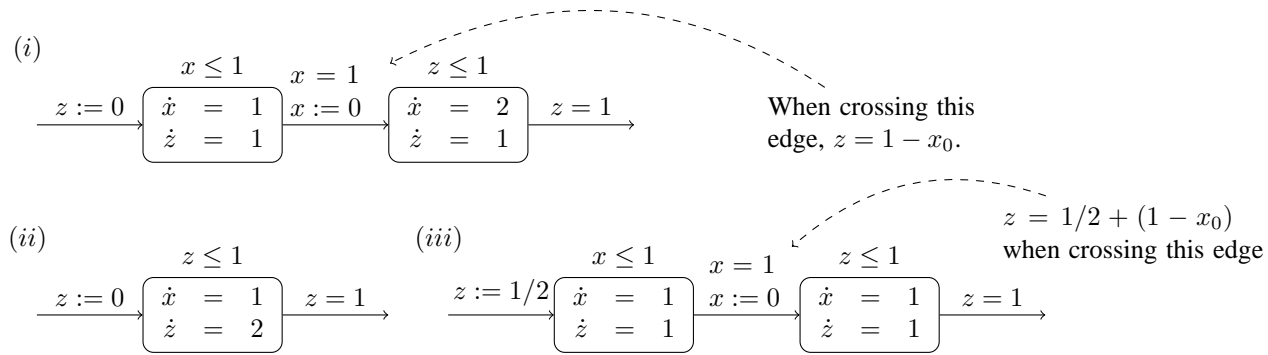




%% file: pract-contrib.tex
\section{Computing fixpoints}
\label{sec:pract-impl}

In this section, we show that Corollary~\ref{cor:theo-for-rha} implies
that we can effectively compute the set of states that are reachable
within $\tb$ time units in an RHA with non-negative rates (using
formulas of the first-order logic $(\reals,0,1,+,\leq)$ over the reals
as a symbolic representation for such sets). We demonstrate, by means
of two examples, that this information can be useful in practice, in
particular when the regular (not time-bounded) fixed points do not
terminate.

\paragraph{\Post\ and \Pre} Let $s$ be state of an RHA with
set of edges $\edges$. Then, we let $\Post(s)\{s'\mid\exists
e\in\edges, t\in\posreal: s\xrightarrow{t,e}s'\}$ and
$\Pre(s)\{s'\mid\exists e\in\edges, t\in\posreal:
s'\xrightarrow{t,e}s\}$. We further let $\reach{\tb}(s)=\{s'\mid \exists
\pi: s\xrightarrow{\pi} s'\land \duration{\pi}\leq \tb\}$, and
$\coreach{\tb}(s)=\{s'\mid \exists \pi: s'\xrightarrow{\pi} s\land
\duration{\pi}\leq \tb\}$ be respectively the set of states that
are reachable from $s$ (that can reach $s$) within $\tb$ time units. We
extend all those operators to sets of states in the obvious
way.

\paragraph{Region algebra}
To symbolically manipulate sets of states, it is well known that we
can use formulas of $(\reals,0,1,+,\leq)$, i.e. the first-order logic
of the reals, with the constants $0$ and $1$, the usual order $\leq$
and addition $+$ (see~\cite{HMR05} for the details).  Recall that the
satisfiability problem for that logic is decidable~\cite{Basu99} and
that it admits effective quantifier elimination. Further remark that,
in a RHA, all guards can be characterised by a formula of
$(\reals,0,1,+,\leq)$ ranging over $X$. Let $\Psi$ be a formula of
$(\reals,0,1,+,\leq)$, and let $\val$ be a valuation of the free
variables of $\Psi$. Then, we write $\val\models\Psi$ iff $\val$
satisfies $\Psi$, and we let $\sem{\Psi}$ be the set off all
valuations $\val$ s.t. $\val\models\Psi$. To emphasise the fact that a
formula $\Psi$ ranges over the set of variables $X$, we sometimes
denote it by $\Psi(X)$.

Based on $(\reals,0,1,+,\leq)$, we can defined a so-called
\emph{algebra of regions}~\cite{HMR05} to effectively represent sets
of states.
The regions\footnote{The notion of region used in this section
  differs from the notion of region given by $\Reg{\cmax,X}$ and used
  to define $\RegHA{\cH'}$. Notice however that any region from
  $\RegHA{\cH'}$ can be expressed via a quantifier free formula of
  $(\reals,0,1,+,\leq)$ with free variables in $X$. The converse is
  obviously not true.} in that algebra can be seen as functions $R$
from the set of locations $\loc$ to quantifier free formula of
$(\reals,0,1,+,\leq)$ with free variables in $X$, representing sets of
valuations for the variables of the RHA. More precisely, any region
$R$ represents the set of states $\sem{R}=\{(\ell,\val)\mid
\val \in \sem{R(\ell)}\}$.
As $(\reals,0,1,+,\leq)$ is closed under all Boolean operations, so is
the region algebra. Since the logic is decidable, testing whether
$s\in\sem{R}$ or whether $\sem{R}=\emptyset$ are both decidable
problems.

In order to obtain fixpoint expressions that characterise
$\reach{\tb}(s)$ and\break $\coreach{\tb}(s)$ using the region algebra, we
introduce $\post$ and $\pre$ operators ranging over regions.  Let $R$
be a region. We let $\post(R)$, be the region s.t. for all $\ell \in
\loc$, $\post(R)(\ell)$ is obtained by eliminating quantifiers in
$\Psi_{\ell}^E(X) \lor \Psi_{\ell}^t(X)$, where $\Psi_{\ell}^E(X)$
characterises all the successors of $R(\ell)$ by an edge with source
$\ell$, and $\Psi_{\ell}^t(X)$ represents all the successors of
$R(\ell)$ by a flow transition in $\ell$ (time elapsing). The
following equations define $\Psi_\ell^E$ and $\Psi_{\ell}^t$, both
ranging on the set of free variables $X$:

\begin{align*}
  &\Psi_\ell^E=\bigvee_{e\in \edges} \psi_\ell^e\\
  &\psi_\ell^{(\ell, g, Y, \ell')} = \exists X' : \left(
    \begin{array}{ll}
      & R(\ell)(X') 
      \land g(X') \\
      \land & \bigwedge_{x \in X \setminus Y} x=x' 
      \land \bigwedge_{x \in Y} x=0 \\
      \land & \invariants(\ell)(X') 
      \land \invariants(\ell')(X)
    \end{array}\right)\\
  &\Psi_\ell^t = \exists t: \exists X' :
  \left(\begin{array}{ll} & t\geq 0 \land R(\ell)(X') \\
      \land & \invariants(\ell)(X) 
      \land  \invariants(\ell)(X') \\
      \land &  \bigwedge_{x \in X}  x' + t \cdot \min(\rate(\ell,x)) \leq x \\  
      \land & \bigwedge_{x \in X}  x \leq x' + t \cdot \max(\rate(\ell,x)) 	
    \end{array}\right)
\end{align*}

Symmetrically, we let $\pre(R)$ be the region s.t. for all $\ell \in
\loc$, $\post(R)(\ell)$ is obtained by eliminating quantifiers in
$\Phi_{\ell}^E(X) \lor \Phi_{\ell}^t(X)$, where $\Phi_{\ell}^E(X)$
represents all the predecessors of $R(\ell)$ by an edge whose target
is $\ell$, and $\Phi_{\ell}^t(X)$ represents all the predecessors of
$R(\ell)$ by a flow transition in $\ell$:
\begin{align*}
  &\Phi_\ell^E =\bigvee_{e\in\edges} \varphi_{\ell}^e\\
  &\varphi_{\ell}^{(\ell, g,Y,\ell')} = \exists X' :
  \left( \begin{array}{ll} 
      & R(\ell')(X') \land g(X) \\
      \land & \bigwedge_{x \in X \setminus Y} x=x' \land \bigwedge_{x \in Y} x'=0 \\
      \land  &\invariants(\ell)(X) \land \invariants(\ell')(X')
    \end{array}\right)\\
  &\Phi_{\ell}^t = \exists t : \exists X' :
  \left(\begin{array}{ll}  & t\geq 0 \land R(\ell)(X') \\
      \land &
      \invariants(\ell)(X)
      \land  \invariants(\ell)(X')\\
      \land & \bigwedge_{x \in X} x + t \cdot \min(\rate(\ell,x)) \leq x' \\
      \land & x' \leq x + t \cdot \max(\rate(\ell,x))
  \end{array}\right)
\end{align*} 

To keep the above definitions compact, we have implicitly assumed that
the rates are given as \emph{closed} rectangles.  The definitions of
$\Phi^t_\ell$ and $\Psi^t_\ell$ can be adapted to cope with intervals
that are left (respectively right) open by substituting $<$ ($>$) for
$\leq$ ($\geq$).

In practice formulas in $(\reals,0,1,+,\leq)$ can be represented and
manipulated as finite union of convex polyhedra for which there exist
efficient implementations, see~\cite{BagnaraHZ08} for example. Those
techniques have been implemented in {\sc HyTech}~\cite{HHW95} and {\sc
  PhaVer}~\cite{Frehse08}. Unfortunately, termination of the symbolic
model-checking algorithms is not ensured for linear hybrid
automata. While in the literature, it is known that forward
reachability and backward reachability fixpoint algorithms terminate
for initialised rectangular hybrid automata~\cite{HenzingerKPV98}, we
show here that termination is also guaranteed for {\em time-bounded
  fixpoint formulas} over the class of \rhadp (that are not
necessarily initialised).

\paragraph{Time-bounded forward and backward fixpoints }
Let $\cH$ be an \rhadp with set of variables $X$, and let $\tb\in\nat$
be a time bound. Let us augment $\cH$ with a fresh variable $t$ to
measure time (hence the rate of $t$ is $1$ in all locations, and $t$
is never reset). Let $R$ be region over the variables $X$. Then, it is
easy to see that the following fixpoint equations characterise
respectively $\reach{\tb}(\sem{R})$ and $\coreach{\tb}(\sem{R})$:
\begin{align}
  \reach{\tb}(\sem{R}) &= \mu Y \cdot ( (\sem{R(X)} \cup \Post(Y))
  \cap \sem{0 \leq t \leq \tb} ) \label{eq:forwardfp}\\
  \coreach{\tb}(\sem{R}) &= \mu Y \cdot ( (\sem{R(X)} \cup \Pre(Y)) \cap
  \sem{0 \leq t \leq \tb} ) \label{eq:backwardfp}
\end{align}

The next lemma ensures that these fixpoints can be effectively
computed. The proof rely on Corollary~\ref{cor:theo-for-rha}.

\begin{lemma}\label{lem:fixpoint} For all \rhadp $\cH$, all region $R$ 
  and all time bound $\tb$, the least fix points (\ref{eq:forwardfp})
  and (\ref{eq:backwardfp}) are respectively equal to the limit of
  $F_0, F_1, F_2,\ldots$ and $B_0, B_1, B_2,\ldots$ where:
  \begin{align*}
    F_0 &=\sem{R(X) \land 1 \leq t \leq \tb}\\
    F_i &=(\Post(F_{i-1}) \cap \sem{0 \leq t \leq
      \tb}) \cup F_{i-1} &\textrm{for all }i>0\\
    B_0 &=\sem{R(X) \land 1 \leq t \leq \tb}\\
    B_i &= B_i=(\Pre(B_{i-1}) \cap \sem{0 \leq t \leq
      \tb}) \cup B_{i-1} &\textrm{for all }i>0\\
  \end{align*}
  Furthermore, both sequences stabilize after at most $F(\cH,\tb)$
  iterations, and both fixpoints can be computed in worst-case doubly
  exponential time.
\end{lemma} 
\begin{proof}
  We justify the result for the least fixpoint
  equation~(\ref{eq:forwardfp}), the result for the least fixpoint
  equation~(\ref{eq:backwardfp}) is justified similarly.

  By induction, it is easy to prove that, for all $i \geq 0$, $F_i$
  contains all the states that are reachable within $\tb$ time units
  \emph{and} by at most $i$ transitions. By
  Corollary.~\ref{cor:theo-for-rha}, we know that all states that
  reachable within $\tb$ time units are reachable by a run of length
  at most $F(\cH,\tb)$. We conclude that
  $F_j=F_{j+1}=\reach{\tb}(\sem{R})$ for $j=F(\cH,\tb)$.  All the
  operations for computing $F_i$ from $F_{i-1}$ take polynomial time
  in the size of $F_{i-1}$, and so the size of $F_i$ is also
  guaranteed to be polynomial in $F_{i-1}$, the overall
  doubly-exponential time bound follows.
\end{proof}
Note that by our {\sc NExpTime}-hardness result, this deterministic
algorithm can be considered optimal (unless {\sc NExpTime}={\sc
  ExpTime}.) Let us now consider two examples to demonstrate that this
approach can be applied in practice.

\input{examples}


%% file: examples.tex
\begin{figure}[ht!]
  \centering
\hrule
\begin{tikzpicture}[node distance=2cm]
   \tikzset{
    state/.style={
           rectangle,
           rounded corners,
           inner sep=2pt,
           text centered,
           draw
           },
         }
  \node (first1) at (0,0) {} ;
  \node[state, label=$0\leq x\leq 1$] (leak) [right=of first1] {%
    $
    \begin{array}{c}
      \textsf{not leaking}\\
      \dot{x}=1\\
      \dot{y}=1\\
      \dot{t}=1
    \end{array}
    $
  } ;

 \node[state, label=$x\geq 0$] (notleak) [right=of leak] {%
    $
    \begin{array}{c}
      \textsf{leaking}\\
      \dot{x}=1\\
      \dot{y}=1\\
      \dot{t}=0
    \end{array}
    $
  } ;

  \path[->] (first1) edge node [above]{$x=0$} (leak)
  (leak) edge [bend left] node [above] {$x:=0$} (notleak)
  (notleak) edge [bend left] node [above] {$x\geq 30$} node [below] {$x:=0$} (leak)
  ;

  \node[state, label={$0\leq x,y\leq 1$}] (l0) [below of= leak, node distance= 2.5cm] {%
    $
    \begin{array}{c}
      \ell_0\\
      \dot{x}=3\\
      \dot{y}=2
    \end{array}
    $
  } ;

 \node[state, label={$0\leq x,y\leq 1$}] (l1) [below of= notleak, node distance=2.5cm] {%
    $
    \begin{array}{c}
      \ell_1\\
      \dot{x}=2\\
      \dot{y}=3
    \end{array}
    $
  } ;

   \node (first2) [left=of l0] {} ;

   \path[->] (first2) edge node [above]{$x=y=0$} (l0)
  (l0) edge [bend left] node [above] {$x=1$} node [below] {$x:=0$} (l1)
  (l1) edge [bend left] node [above] {$y=1$} node [below] {$y:=0$} (l0)
  ;
\end{tikzpicture}

\hrule
\caption{A stopwatch automaton for the leaking gas burner (top) and an
  SHA with bounded invariants (bottom).}
\label{fig:examples}
\end{figure}
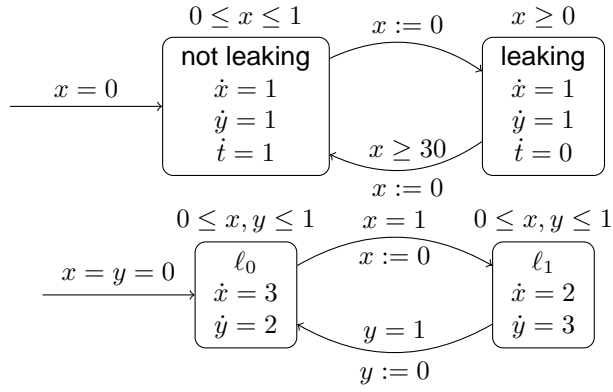

\paragraph{Example 1: Leaking gas burner}
We present an example of a system where the classical fixpoint
computation for reachability analysis does not terminate, while the
time-bounded analysis does terminate.  Consider the example of a
leaking gas burner~\cite{ACHH93}.  The gas burner can be either
\emph{leaking} or \emph{not leaking}.  Leakages are repaired within
1~second, and no leakage can happen in the next 30~seconds after a
repair.  In Fig.~\ref{fig:examples} (top), an automaton with two
locations and the clock~$x$ is a model of the gas burner. In order to
measure the leakage time and the total elapsed time, the stopwatch~$t$
and clock~$y$ are used as monitors of the system. It was shown using
backward reachability analysis that in any time interval of at least
60~seconds, the time of leakage is at most one twentieth of the
elapsed time~\cite{HHW95}.  The fixpoint is computed after
$7$~iterations of the backward reachability algorithm. However, the
forward reachability analysis does not terminate.

Using forward time-bounded reachability analysis we can prove the property that 
in all time intervals of fixed length $T \geq 60$, the leakage time is at most
$\frac{T}{20}$. 
In order to prove that this property holds in \emph{all} time intervals, 
we perform the reachability analysis from all possible states of the system
(i.e., from location {\sf leaking} with $0 \leq x\leq 1$, and location 
{\sf not\_leaking} with $x \geq 0$) and starting with $t=y=0$. 
For a fixed time bound $T$, we compute the set of reachable states satisfying 
$y \leq T$ and check that $t \leq \frac{T}{20}$ when $y = T$. 
The results of this paper guarantees that the analysis terminates.
Using HyTech, the property is established for $T=60$ after $5$~iterations of 
the forward time-bounded fixpoint algorithm.
Thus for all time intervals of $T=60$~seconds, 
the leakage time is at most $\frac{T}{20}$.



\paragraph{Example 2: bounded invariant}
In Fig.~\ref{fig:examples} (bottom), we consider a rectangular
automaton with positive rates where all variables have a bounded
invariant~$[0,1]$.  In this example, the forward reachability analysis
of HyTech does not terminate because the set of reachable states is
not a finite union of polyhedra (see Fig.~\ref{fig:reach}).  On the
other hand, the time-bounded forward fixpoint terminates by
Lemma~\ref{lem:fixpoint}.  This example shows that it is not
sufficient to bound the variables in the automaton to get termination,
but it is necessary to bound the time horizon of the analysis.

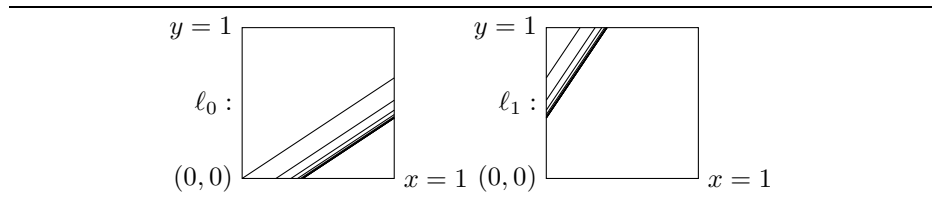
\begin{figure}[ht!]
\centering
\hrule
\begin{tikzpicture}[scale=2]
  \foreach \x/\y in {0/0,2/1} {
    \draw (\x,0) rectangle +(1,1) ;
    \node at (\x,0) [left] {$(0,0)$} ;
    \node at (\x,1) [left] {$y=1$} ;
    \node at (\x+1,0) [right] {$x=1$} ;
    \node at (\x, .5) [left] {$\ell_{\y}:$} ;
  }
 
  \xdef\v{0}
  \foreach \count in {1,...,10} {
    \xdef\lastv{\v}
    \pgfmathparse{(1-\lastv)*2/3}
    \xdef\v{\pgfmathresult}
    \draw (\lastv, 0) -- (1, \v) ;

    \xdef\lastv{\v}
    \pgfmathparse{(1-\lastv)*2/3}
    \xdef\v{\pgfmathresult}
    \draw (2, \lastv) -- (\v+2, 1) ;
}

\end{tikzpicture}

\hrule
\caption{{\bf Reachable states for the automaton of
    Fig.~\ref{fig:examples} (bottom).}
\label{fig:reach}}
\end{figure}


%% file: appendix.tex
\subsection*{Proof of Lemma~\ref{lemma:reg-ha-preserves-reach}}

\medskip

\input{prooflemma-reg-ha-pr-reach}

\subsection*{Proof of Lemma~\ref{lemma:value-variables}}

\medskip

\begin{proof}
  Let us first prove that, for all $0\leq i\leq n$: $\val_i\in
  r_i$. The proof is by induction on $i$. For $i=0$, the property
  holds by definition of $\RegHA{\cH'}$ and because $((\ell_0,
  r_0),\val_0)$ is an initial state. Assume $\val_i\in r_i$
  for some $i\geq 0$, and let us show that $\val_{i+1}\in
  r_{i+1}$. Let $g$ and $Y$ denote respectively the guard and the
  reset set of $e_{i+1}$. Let
  $\val'=\val_i+t_{i+1}\times\rate(\ell_i)$ be the valuation of
  the variables when crossing $e_{i+1}$. By construction, we know that
  there is a region $r''$ which is a conjunct of $g$ (hence $\val'\in
  r''$) s.t. for all $x\not\in Y$: $r''(x)=r_{i+1}(x)$. Thus, for all
  $x\not\in Y$: $\val'(x)=\val_{i+1}(x)\in r_{i+1}(x)$. Moreover,
  still by construction, for all $x\in Y$: $r_{i+1}(x)\in \{\zeq ,
  \zplus\}$. Hence, for all $x\in Y$: $\val_{i+1}(x)=0\in r_{i+1}$
  too.

  To conclude, the two last points of the lemma follow immediately
  from the construction of $\RegHA{\cH'}$, as for all edge $e$ and all
  variable $x$ s.t.  $\src{e}=(\ell,r)$ and $r(x)=\zeq $
  (resp. $r(x)=\zplus$), the constraint $x=0$ ($x>0$) appears as a
  conjunct of $e$'s guard.
\end{proof}


%% file: prooflemma-reg-ha-pr-reach.tex
\begin{proof}
  With each run
  $\rho=(\ell_0,\val_0),(t_1,e_1),(\ell_1,\val_1),\ldots,\break
  (t_n,e_n),(\ell_n,\val_n)$ of $\cH'$, we associate a run
  $\rho'=((\ell_0,r_0),\val_0),\break
  (t_0,e_0'),((\ell_1,r_1),\val_1),\ldots,
  (t_n,e_n'),((\ell_n,r_n),\val_n)$ of $\RegHA{\cH'}$ (hence with
  $\duration{\rho}=\duration{\rho'}$) s.t.:
  \begin{itemize}
  \item for all $0\leq i\leq n$, for all $x\in X$ (assuming
    $t_{n+1}=0$):
\[
      r_i(x) =
      \begin{cases}
        \zplus&\textrm{if }\val_i(x)=0\textrm{ and }(t_{i+1}>0\textrm{
          and }\dot{x}\neq 0)\\
        \zeq&\textrm{if }\val_i(x)=0\textrm{ and }(t_{i+1}=0\textrm{
          or }\dot{x}= 0)\\
        \regof{\val_i(x)}&\textrm{otherwise}
      \end{cases}
\]


  \item $e_i'$ is the unique edge between $(\ell_i,r_i)$ and
    $(\ell_{i+1},r_{i+1})$, corresponding to $e_i$.
  \end{itemize}
We prove by (backward) induction (on $n$) that the run $\rho'$ is a
genuine run of $\RegHA{\cH'}$. When $n=0$, there is nothing to
prove. Let us now assume that given
$\rho=(\ell_0,\val_0),\break (t_1,e_1),(\ell_1,\val_1),\ldots,
(t_n,e_n),\linebreak(\ell_n,\val_n)$ run of $\cH'$, we have proved
that $((\ell_1,r_1),\val_1),\ldots,
(t_n,e_n'),\linebreak((\ell_n,r_n),\val_n)$ is a genuine of
$\RegHA{\cH'}$. To obtain the desired result, it remains to prove that
$((\ell_0,r_0),\val_0), (t_0,e_0'),((\ell_1,r_1),\val_1)$ is a genuine
run of $\RegHA{\cH'}$. For this, we have to prove that for all $x \in
X$: \emph{(i)} $\val_0(x)+ t \times\rate(\ell_0,r_0)(x) \models
\invariants(\ell_0,r_0)$, for all $0 \le t \le t_1$, \emph{(ii)}
$\val_0(x)+ t_1 \times\rate(\ell_0,r_0)(x) \models g'$ (where $g'$ is
the guard of the transition $e'_1$), and \emph{(iii)} $\val_1(x)=0$
(resp. $\val_1(x)=\val_0(x)$) if $x \in Y'$ (resp. $x \notin Y'$)
(where $Y'$ is the reset of the transition $e'_1$). Let us distinguish
three cases:
\begin{enumerate}
\item Case~1: $r(x) = \zplus$. In this case, by construction of
  $\rho'$, we know that $\val_0(x)=0$, $t_1 > 0$ and $\dot{x} \ne
  0$. In particular, we have that $\val_0(x)+ t_1
  \times\rate(\ell_0,r_0)(x)>0$. By construction of $\RegHA{\cH'}$, we
  know that $g'(x)=g(x) \wedge r''(x) \wedge (x>0)$, where $g(x)$
  (resp. $g'(x)$) represents the constraints\footnote{Notice that it
    makes sense to decouple guard according to variables since there
    are no diagonal constraints.}  on $x$ in the guard of the
  transition $e_1$ (resp. $e'_1$), and $r\tsucc{\rate(\ell_0,r_0)}
  r''$. Moreover we have that
  $\invariants(\ell_0,r_0)(x)=\invariants(\ell_0)(x)$.  Since $\rho$
  is a genuine run of $\cH'$, we clearly have that \emph{(i)} and
  \emph{(iii)} are satisfied. Point \emph{(ii)} follows from the facts
  that $\rho$ is a genuine run of $\cH'$ and that $\val_0(x)+ t_1
  \times\rate(\ell_0,r_0)(x)>0$.
\item Case~2: $r(x) = \zeq$. In this case, by construction of $\rho'$,
  we know that $\val_0(x)=0$, $t_1 = 0$ or $\dot{x} = 0$. In
  particular, we have that $\val_0(x)+ t_1
  \times\rate(\ell_0,r_0)(x)=0$. By construction of $\RegHA{\cH'}$, we
  know that $g'(x)=g(x) \wedge r''(x) \wedge (x=0)$, where $g(x)$
  (resp. $g'(x)$) represents the constraints on $x$ in the guard of
  the transition $e_1$ (resp. $e'_1$), and $r\tsucc{\rate(\ell_0,r_0)}
  r''$. Moreover we have that
  $\invariants(\ell_0,r_0)(x)=\invariants(\ell_0)(x) \wedge (x=0)$.
  Since $\rho$ is a genuine run of $\cH'$, we clearly have that
  \emph{(iii)} is satisfied. Points \emph{(i)} and \emph{(ii)} follow
  from the facts that $\rho$ is a genuine run of $\cH'$ and that
  $\val_0(x)+ t \times\rate(\ell_0,r_0)(x)=0$, for all $0 \le t \le
  t_1$. 
  \item Case~3: $r(x) \notin \{\zeq,\zplus\}$. This case is simpler
    than the two previous ones. The three points \emph{(i)},
    \emph{(ii)} and \emph{(iii)} follow from the facts that $\rho$ is
    a genuine run of $\cH'$. 
\end{enumerate}

Then, with each run
$\rho'=((\ell_0,r_0),\val_0),(t_0,e_0'),((\ell_1,r_1),\val_1),\ldots,\break
(t_n,e_n'), ((\ell_n,r_n),\val_n)$ of $\RegHA{\cH'}$ we associate the
run $\rho=(\ell_0,\val_0),\break (t_1,e_1),(\ell_1,\val_1),\ldots,
(t_n,e_n), (\ell_n,\val_n)$ where, for all $1\leq i\leq n$, $e_i$ is
the unique edge of $\cH'$ that corresponds to $e_i'$. Since the guards
and invariant of $\RegHA{\cH'}$ are more
constraining 
than those of $\cH'$, the fact that $\rho'$ is a genuine run of
$\RegHA{\cH'}$ implies that $\rho$ is a genuine run of
$\cH'$. 
 \end{proof}


%% file: main.bbl
\begin{thebibliography}{10}

\bibitem{ACHH93}
R.~Alur, C.~Courcoubetis, T.~A. Henzinger, and P.-H. Ho.
\newblock Hybrid automata: {A}n algorithmic approach to the specification and
  verification of hybrid systems.
\newblock In {\em Hybrid Systems}, LNCS 736, pages 209--229. Springer, 1993.

\bibitem{AD94}
R.~Alur and D.~Dill.
\newblock A theory of timed automata.
\newblock {\em Theoretical Computer Science}, 126(2):183--235, 1994.

\bibitem{BagnaraHZ08}
R.~Bagnara, P.~M. Hill, and E.~Zaffanella.
\newblock The parma polyhedra library: Toward a complete set of numerical
  abstractions for the analysis and verification of hardware and software
  systems.
\newblock {\em Sci. Comput. Program.}, 72(1-2):3--21, 2008.

\bibitem{Basu99}
S.~Basu.
\newblock New results on quantifier elimination over real closed fields and
  applications to constraint databases.
\newblock {\em J. ACM}, 46(4):537--555, 1999.

\bibitem{ICALP11}
T.~Brihaye, L.~Doyen, G.~Geeraerts, J.~Ouaknine, J.-F. Raskin, and J.~Worrell.
\newblock On reachability for hybrid automata over bounded time.
\newblock In {\em ICALP (2)}, volume 6756 of {\em Lecture Notes in Computer
  Science}, pages 416--427. Springer, 2011.

\bibitem{Frehse08}
G.~Frehse.
\newblock Phaver: algorithmic verification of hybrid systems past hytech.
\newblock {\em STTT}, 10(3):263--279, 2008.

\bibitem{DBLP:conf/lics/Henzinger96}
T.~A. Henzinger.
\newblock The theory of hybrid automata.
\newblock In {\em LICS}, pages 278--292. IEEE Computer Society, 1996.

\bibitem{HHW95}
T.~A. Henzinger, P.-H. Ho, and H.~Wong-Toi.
\newblock A user guide to {{\sc HyTech}}.
\newblock In {\em TACAS'95: Tools and Algorithms for the Construction and
  Analysis of Systems}, volume 1019 of {\em Lecture Notes in Computer Science},
  pages 41--71. Springer-Verlag, 1995.

\bibitem{DBLP:conf/cav/HenzingerHW97}
T.~A. Henzinger, P.-H. Ho, and H.~Wong-Toi.
\newblock Hytech: A model checker for hybrid systems.
\newblock In {\em CAV}, volume 1254 of {\em Lecture Notes in Computer Science},
  pages 460--463. Springer, 1997.

\bibitem{HenzingerKPV98}
T.~A. Henzinger, P.~W. Kopke, A.~Puri, and P.~Varaiya.
\newblock What's decidable about hybrid automata.
\newblock {\em Journal of Computer and System Sciences}, 57(1):94--124, 1998.

\bibitem{HMR05}
T.~A. Henzinger, R.~Majumdar, and J.-F. Raskin.
\newblock A classification of symbolic transition systems.
\newblock {\em ACM Trans. Comput. Log.}, 6(1):1--32, 2005.

\bibitem{DBLP:conf/lics/JenkinsORW10}
M.~Jenkins, J.~Ouaknine, A.~Rabinovich, and J.~Worrell.
\newblock Alternating timed automata over bounded time.
\newblock In {\em LICS}, pages 60--69. IEEE Computer Society, 2010.

\bibitem{DBLP:conf/hybrid/JhaKWC07}
S.~K. Jha, B.~H. Krogh, J.~E. Weimer, and E.~M. Clarke.
\newblock Reachability for linear hybrid automata using iterative relaxation
  abstraction.
\newblock In {\em HSCC}, volume 4416 of {\em Lecture Notes in Computer
  Science}, pages 287--300. Springer, 2007.

\bibitem{DBLP:conf/concur/OuaknineRW09}
J.~Ouaknine, A.~Rabinovich, and J.~Worrell.
\newblock Time-bounded verification.
\newblock In {\em CONCUR}, volume 5710 of {\em Lecture Notes in Computer
  Science}, pages 496--510. Springer, 2009.

\end{thebibliography}
